\documentclass[a4paper,11pt]{article}
\usepackage[margin=2cm]{geometry}
\usepackage{amsmath,amssymb,amsfonts,amsthm,float,ulem}

\usepackage{tikz}

\def\HE{\text{H1}}
\def\HT{\text{H2}}
\def\HD{\text{H3}}
\def\AE{\text{A1}}
\def\AT{\text{A2}}
\def\QE{\text{Q1}}
\def\QT{\text{Q2}}
\def\QD{\text{Q3}}
\def\QV{\text{Q4}}

\def\t#1{\widetilde{#1}}
\def\h#1{\widehat{#1}}
\def\th#1{\widehat{\widetilde{#1}}}
\def\b#1{\overline{#1}}

\def\uh#1{\underset{\text{$\hat{\phantom{\cdot}}$}}{#1}}

\newcounter{NN}
\usepackage{xcolor}

\newtheorem{remark}[NN]{Remark}

\newtheorem{theorem}[NN]{Theorem}

\newtheorem{lemma}[NN]{Lemma}

\begin{document}
\title{From auto-B\"acklund transformations to auto-B\"acklund transformations, and torqued ABS equations}
\author{
Dan-da Zhang$^1$, Da-jun Zhang$^2$, Peter H.~van der Kamp$^3$\\[2mm]
$^1$School of Mathematics and Statistics, Ningbo University, Ningbo 315211, China\\
zhangdanda@nbu.edu.cn,
https://orcid.org/0000-0001-6406-6672\\
$^2$Department of Mathematics, Shanghai University, Shanghai 200444, China\\
djzhang@staff.shu.edu.cn,
https://orcid.org/0000-0003-3691-4165\\
$^3$Department of Mathematics and Statistics, La Trobe University, Victoria 3086, Australia\\
P.vanderKamp@LaTrobe.edu.au,
https://orcid.org/0000-0002-2963-3528}

\date{}

\maketitle

\begin{abstract}
We provide a method which from a given auto-B\"acklund transformation (auto-BT) produces another auto-BT for a different equation. We apply the method to the natural auto-BTs for the ABS quad equations, which gives rise to torqued versions of ABS equations and explains the origin of each auto-BT listed in [J. Atkinson, J. Phys. A: Math. Theor. 41 (2008) 135202]. The method is also applied to non-natural auto-BTs for ABS equations, which yields 3D consistent cubes which have not been found in [R. Boll, J. Nonl. Math. Phys. 18 (2011) 337--365], and to a multi-quadratic ABS* equation giving rise to a multi-quartic equation. \\[1mm]
MSC class: 37K60
\end{abstract}

\section{Introduction}
Consistency Around the Cube (CAC) is a central notion in the study of discrete integrable systems.
Three quadrilateral equations $A=B=Q=0$\footnote{This is short-hand notation for the system $\{A=0, B=0, Q=0\}$,
which is complemented by the equations on the opposite faces $\{\h{A}=0, \t{B}=0, \overline{Q}=0\}$,
where tilde and hat denote shifts in different directions, e.g.
if $u=u(n,m)$ and $A=f(u)$, then $\uh{\t{A}} = f(u(n+1,m-1))$, and $\bar{u}=v$.}, each posed on two opposite faces of a cube, as in Figure \ref{1a},
is called CAC (or 3D consistent) if values can be assigned to $u,v$ and their shifts, such that each equation is satisfied \cite{ABS03,Atk08,AtkN-IMRN-2014,LY,Nij02,NW01,XP}. For multi-linear equations the property is usually characterised as an initial value problem, where $\th v$ is determined uniquely from $u,\t u, \h u, v$. It is possible to define different equations on opposite faces of the cube \cite{RB11,ZVZ}, however we will not consider that situation in this paper.

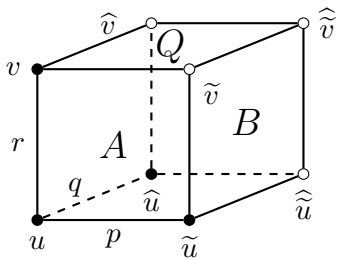
\begin{figure}[h]
	\centering
	\begin{tikzpicture}[scale=.5]
	\tikzstyle{nod1}= [circle, inner sep=0pt, fill=white, minimum size=4pt, draw]
	\tikzstyle{nod}= [circle, inner sep=0pt, fill=black, minimum size=4pt, draw]
	\def\lx{3}%
	\def\ly{1.22}%
	\def\l{4}%
	\def\d{4}%
\node[nod] (u00) at (0,0) [label=below: $u$] {};
	\node[nod] (u10) at (\l,0) [label=below: $\widetilde{u}$] {};
	\node[nod] (u01) at  (\lx,\ly) [label=below: $\widehat{u}$] {};
	\node[nod1] (u11) at (\l+\lx,\ly) [label=below: $\widehat{\widetilde{u}}$] {};
	\node[nod] (v00) at (0,\d) [label=left: $v$] {};
	\node[nod1] (v10) at (\l,\d) [label=below right: $\widetilde{v}$] {};
	\node[nod1] (v01) at (\lx,\d+\ly) [label=left: $\widehat{v}\ \ $] {};
	\node[nod1] (v11) at (\l+\lx,\d+\ly) [label=right: $\widehat{\widetilde{v}}$] {};

\draw (\l/2,-.5) node {$p$};
\draw (-.5,\d/2) node {$r$};
\draw (\lx/2-.5,\ly/2+.25) node {$q$};

\draw[font=\Large]  (\l/2,\d/2) node {$A$};
\draw[font=\Large]  (\l+\lx/2,\d/2+\ly/2) node {$B$};
\draw[font=\Large]  (\l/2+\lx/2,\d+\ly/2) node {$Q$};

	\draw[thick]  (u00) -- (u10) -- (u11) -- (v11) -- (v10);
	\draw[thick]  (u00) -- (v00) -- (v01) -- (v11);
	\draw[thick]  (v00) -- (v10) -- (u10);
	\draw[thick,dashed]  (u00)  -- (u01) -- (u11);
	\draw[thick,dashed]  (u01) -- (v01);
	\end{tikzpicture}
\caption{Consistency: quad equations can be posed on the six faces of a cube,
we have $A=0$ on the front and back faces, $B=0$ on the left and right faces,
and $Q=0$ on the top and bottom faces. The black dots indicate initial values.
Each equation depends on two of the three lattice parameters $p,q,r$.} \label{1a}
\end{figure}

B\"acklund transformations (BTs) and their superposition principle were introduced by B\"acklund \cite{Bac} and Bianchi \cite{Bia} as transformations of pseudospherical surfaces, which are solutions of the (partial differential) sine-Gordon equation. These transformations constitute the remarkable connections that exist between the classical differential geometry of surfaces and modern soliton theory \cite{RS}. In \cite{L,LB}, the relation between nonlinear differential-difference equations associated with discrete Schr\"odinger and Zakharov-Shabat spectral problems and BTs for nonlinear evolution equations was revealed. In \cite{GY1}, lattice equations were constructed as auto-BTs for Volterra and Toda differential-difference equations. Wahlquist and Estabrook \cite{WE} showed that by recursive application of the BT to any solution of the Korteweg-de Vries equation generates a hierarchy of solutions which satisfy a superposition principle. The connection between Bianchi permutability (for BTs for nonlinear evolution equations)
and lattice equations was first established in \cite{NQC}. In this paper, the setting is fully discrete, i.e. a system of lattice equations acts as an auto-BT for a lattice equation, which is closely connected to CAC.

The system $A=B=0$ (together with the upshifted equations $\h A=\t B=0$)
is an auto-B\"acklund transformation (auto-BT) for the equation $Q=0$. Hietarinta \cite[Definitions 5.1,5.2]{JH19} distinguishes two types: eqBT and solBT. An auto-BT is a BT of solutions (solBT) if from a given solution $u$ of $Q=0$ one can solve the auto-BT $A=B=0$ to find another solution $v$ of $Q=0$. The notion of B\"acklund transform of equations (eqBT) is what we adopt here; if the auto-BT $A=B=0$ (through consistency) gives rise to $Q=0$ uniquely, the auto-BT is called a {\bf strong} eqBT. If it gives rise to $Q=0$ but not uniquely, the auto-BT is called a {\bf weak} eqBT, cf. \cite{HV}. And, if it does not give rise to $Q=0$, it is called {\bf trivial} or fake, cf. \cite{MH, BH}. Usually, an auto-BT depends on a parameter, called the B\"acklund parameter, and solutions related by auto-BTs satisfy a superposition principle (Bianchi permutability) \cite{Atk08, RB13}.

In the more special situation where the three equations have the same form, i.e. $Q=Q(p,q)$, $A=Q(p,r)$, $B=Q(q,r)$, the equation $Q=0$ is called CAC. Such equations provide their own, natural auto-BT. Multi-affine quadrilateral equations have been classified with respect to CAC in \cite{ABS03}. Under the additional assumptions, that the equations are D4-symmetric
and possess the tetrahedron property, i.e. a relation exists between $\t u$, $\h u$, $v$ and $\th v$, Adler, Bobenko and Suris (ABS) obtained a list of 9 equations, three equations of type H, two equations of type A, and 4 equations of type Q. For completeness, the list is included in section 3. Multi-quadratic equations, denoted ABS$^\ast$, with the CAC property were given in \cite{AtkN-IMRN-2014}. The ABS$^\ast$ equations define multi-valued evolution, however, due to the discriminant factorization property, they allow reformulation as a single-valued system, and possess BTs to the ABS equations. Lists of BTs between ABS equations, as well as non-natural auto-BTs (each with a superposition principle) were given in \cite{Atk08}. Auto-BTs play an important role in the construction of solutions, cf. \cite{HZ10,HZ09,HZ11,WVZ}.

In 2009, Adler, Bobenko and Suris classified multi-affine cube systems (with a priori different equations on each face) without the assumptions of D4-symmetry and the tetrahedron property, but with an additional non-degeneracy condition. Their main result \cite[Theorem 4]{ABS09} states that each 3D-consistent system of type Q is, up to M\"{o}bius transformations, one of the known Q-type systems (with D4-symmetry) found in \cite{ABS03}. This includes the equations of type A, which are related to equations $\QE$ and $\QD_0$ by (non-autonomous) point transformations. A classification of the degenerate cases (the H-equations) has been performed in \cite{Boll,RB11}. To each of the 6 edges (including the diagonals) of a quad equation one associates a biquadratic, and equations with $i$ degenerate biquadratics, $i>0$, are said to be of type H$^i$.
There are 3 quad equations of type H$^4$, 4 quad equations of type H$^6$, and a list of 3D consistent systems (with the tetrahedron property) was given in \cite[Theorem 3.4]{RB11}.

Inspired by the work \cite{ZVZ}, where multi-component generalisations of CAC lattice systems were obtained by extending the scalar variable to a diagonal matrix and applying cyclic transformations to shifted matrices, we consider the equations of the auto-BT as dependent on two 2-component vectors, and apply cyclic transformations to those. This provides a systematic method (which can be called torquing) which takes an auto-BT (invariant under interchanging $u\leftrightarrow v$) and produces another auto-BT for a different equation. In the symmetric case, the new equation is related to the tetrahedron property of the original system. In this case, the result is similar to \cite[Theorem 3.3]{RB11} in the more general setting where equations on opposite faces of the cube do not have to coincide. In the asymmetric case, the new equation relates to a (previously hidden) relation, which can be obtained from the auto-BT, satisfying a property we have called planarity.

The paper is organized as follows. In section 2, we observe that for an ABS-equation $Q(u)=0$ the natural auto-BT is not only an auto-BT for $Q(u)$, but also for an equation $R(u,v)=0$, which depends on both $u,v$. We define the notion of planarity for such equations. We prove a useful result, which implies that one can torque an auto-BT in two different ways (symmetric and non-symmetric). Full details are provided for $\HE$, where an additive\footnote{Because B\"acklund parameters in ABS equations are usually connected to addition on different curves: lines (the usual additive case), exponential functions (the usual multiplicative case), and elliptic curves (the elliptic case) \cite{ABS03,RB11}, we point out that here there is no relation to addition on curves. The terms additive resp. multiplicative refer to the additive resp. multiplicative nature of our transformations of lattice parameters, (\ref{pp}) and (\ref{ps}).} transformation of the lattice parameters is required. A pictorial representation of the consistent cubes on which symmetric and asymmetric torqued auto-BTs, as well as their superposition principles, is also provided. In section 3, we provide the results of torquing the natural auto-BTs for all ABS equations. For 6 of them ($\HE$, $\HT$, $\AE_\delta$, $\QE_\delta$, $\QT$, $\QV$)
the same additive parameter transformation applies. For the remaining 3 ABS equations ($\HD_\delta$, $\AT$, $\QD_\delta$) a multiplicative transformation of the lattice parameters is required. Details are provided for $\HD_\delta$. We obtain, amongst other cube systems, all auto-BTs listed in \cite[Table 2]{Atk08}, together with their superposition principles.
In section 5, we apply torquing to the non-natural auto-BTs for $\HD_\delta$, $\AT$ and $\QD_\delta$ found in \cite{ZZ}. The corresponding consistent cube systems, and their torqued versions seem to missing in the classification of \cite[Theorem 3.3]{RB11}. In section 6, we show that the multi-quadratic model $\HD^*_0$ \cite[Equation (22)]{AtkN-IMRN-2014} can also be torqued, giving rise to a multi-quartic equation. In the final section, we summarise our findings and mention some related results.

\section{Torqued equations}\label{sec-2}
In this paper, we will write the dependence on the field values, and the parameters, as
\begin{equation}\label{Q}
Q(u)=Q([u,\h u],[\t u,\th u];\ p,q)=0,
\end{equation}
instead of the more standard $Q(u,\t u,\h u,\th u;\ p,q)=0$. We assume equations on opposite faces of the cube are identical,
and that the auto-BT
\begin{equation}\label{AB}
A=Q([u,v],[\t u,\t v];\ p,r)=0,\qquad B=Q([u,v],[\h u,\h v];\ q,r)=0,
\end{equation}
which we denote shortly by $A=B=0$, is symmetric under interchanging $u\leftrightarrow v$.

In this section, we lay out and illustrate our method using the equation H1:
\[
\HE([u,\h u],[\t u, \th u];\ p, q):=(u-\th u)(\t u-\h u)-p+q=0.
\]
Given a non-trivial auto-BT, one is able to determine the equation for which it is an auto-BT. As we shall see, in section \ref{fabte}, it may actually be an auto-BT for more than one equation.

\subsection{From an auto-BT to the equations it is an auto-BT for} \label{fabte}
We start from a consistent cube equipped with $A=0$ on the front face, $\h A=0$ on the back face,
and $B=0$ and $\t B=0$ on the side faces, but the two equations $Q(u)=0$, $Q(v)=0$ on the bottom and top faces are omitted. Taking $u,\t u,\h u$ and $v$ as initial values, as in Figure \ref{1a},
the values of $\t v,\h v$ are uniquely determined by the equations $A=B=0$ (\ref{AB}). We then solve the system of equations $\h A=\t B=0$ for $\th u, \th v$. For multi-affine equations,  if the auto-BT is not trivial, this gives rise to two solutions, one of which will satisfy the decoupled set of equations
\begin{equation} \label{QQ}
Q(u)=Q(v)=0.
\end{equation}
In the above procedure, the equation $Q(u)=0$ is obtained directly. However, to obtain $Q(v)=0$ one needs to substitute the solution of $A=B=0$ with respect to $\t u,\h u$. By doing so the dependence on $u$ disappears.

If the auto-BT is not weak, the second solution corresponds to a coupled set of equations,
which we denote, assuming that the auto-BT is $u\leftrightarrow v$ symmetric, by
\begin{equation} \label{RR}
R(u,v)=R(v,u)=0.
\end{equation}
Here, the equation $R(u,v)=0$ is the equation for $\th v$, where we have substituted the solution of $A=B=0$ with respect to $\t u,\h v$, and $\overline{R}(u,v)=R(v,u)=0$ is the equation for $\th u$, having used the solution of $A=B=0$ with respect to $\h u,\t v$. If $R(u,v)=0$ depends on $u,\h u,\t v, \th v$ but not on $v$, and $R(v,u)=0$ depends on $v,\h v,\t u, \th u$ but not on $u$, then the equations \eqref{RR} are called {\bf planar}, cf. Figure \ref{planar}. According to \cite[Definitions 5.1]{JH19} the equation $R(u,v)=0$ is not an `acceptable' equation. The procedure of torquing, which we will soon embark upon, turns it into an acceptable equation.

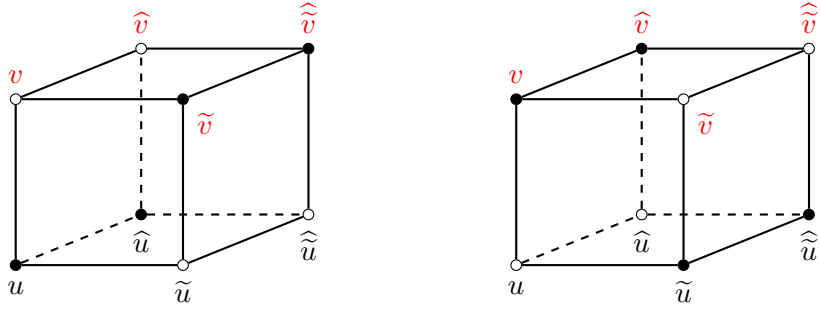
\begin{figure}[H]
	\centering
	\begin{tikzpicture}[scale=.55]
	\tikzstyle{nod1}= [circle, inner sep=0pt, fill=white, minimum size=4pt, draw]
	\tikzstyle{nod}= [circle, inner sep=0pt, fill=black, minimum size=4pt, draw]
	\def\lx{3}%
	\def\ly{1.22}%
	\def\l{4}%
	\def\d{4}%
\node[nod] (u00) at (0,0) [label=below: $u$] {};
	\node[nod1] (u10) at (\l,0) [label=below: $\widetilde{u}$] {};
	\node[nod] (u01) at  (\lx,\ly) [label=below: $\widehat{u}$] {};
	\node[nod1] (u11) at (\l+\lx,\ly) [label=below: $\widehat{\widetilde{u}}$] {};
	\node[nod1] (v00) at (0,\d) [label=above: {\color{red}$v$}] {};
	\node[nod] (v10) at (\l,\d) [label=below right: {\color{red}$\widetilde{v}$}] {};
	\node[nod1] (v01) at (\lx,\d+\ly) [label=above: {\color{red}$\widehat{v}$}] {};
	\node[nod] (v11) at (\l+\lx,\d+\ly) [label=above: {\color{red}$\widehat{\widetilde{v}}$}] {};

	\draw[thick]  (u00) -- (u10) -- (u11) -- (v11) -- (v10);
	\draw[thick]  (u00) -- (v00) -- (v01) -- (v11);
	\draw[thick]  (v00) -- (v10) -- (u10);
	\draw[thick,dashed]  (u00)  -- (u01) -- (u11);
	\draw[thick,dashed]  (u01) -- (v01);
	\end{tikzpicture}
\hspace{2cm}
	\begin{tikzpicture}[scale=.55]
	\tikzstyle{nod1}= [circle, inner sep=0pt, fill=white, minimum size=4pt, draw]
	\tikzstyle{nod}= [circle, inner sep=0pt, fill=black, minimum size=4pt, draw]
	\def\lx{3}%
	\def\ly{1.22}%
	\def\l{4}%
	\def\d{4}%
\node[nod1] (u00) at (0,0) [label=below: $u$] {};
	\node[nod] (u10) at (\l,0) [label=below: $\widetilde{u}$] {};
	\node[nod1] (u01) at  (\lx,\ly) [label=below: $\widehat{u}$] {};
	\node[nod] (u11) at (\l+\lx,\ly) [label=below: $\widehat{\widetilde{u}}$] {};
	\node[nod] (v00) at (0,\d) [label=above: {\color{red}$v$}] {};
	\node[nod1] (v10) at (\l,\d) [label=below right: {\color{red}$\widetilde{v}$}] {};
	\node[nod] (v01) at (\lx,\d+\ly) [label=above: {\color{red}$\widehat{v}$}] {};
	\node[nod1] (v11) at (\l+\lx,\d+\ly) [label=above: {\color{red}$\widehat{\widetilde{v}}$}] {};
	\draw[thick]  (u00) -- (u10) -- (u11) -- (v11) -- (v10);
	\draw[thick]  (u00) -- (v00) -- (v01) -- (v11);
	\draw[thick]  (v00) -- (v10) -- (u10);
	\draw[thick,dashed]  (u00)  -- (u01) -- (u11);
	\draw[thick,dashed]  (u01) -- (v01);
	\end{tikzpicture}
\caption{\label{planar} The stencils for a set of planar equations $R(u,v)=0$ (left)
and $R(v,u)=0$ (right) are indicated by the black dots. Each equation depends on $u$-variables (black) and $v$-variables (red).}.
\end{figure}

Let us now revisit the first solution (\ref{QQ}). By substituting the solution of $A=B=0$ with respect to $\t u,\h u$
into equation $Q(u)=0$ in \eqref{QQ} we get an equation $T(u,v)=0$ which depends on $u, \t v,\h v, \th u$ and possibly $v$. Similarly, an equation
$\overline{T}(u,v)=T(v,u)=0$ is obtained by substituting the solution of $A=B=0$ with respect to $\t v,\h v$ into equation $Q(v)=0$.
Equations obtained in this way,
\begin{equation}\label{TT}
T(u,v)=0,\ T(v,u)=0
\end{equation}
are said to have the {\bf tetrahedron} property if $T(u,v)$ does not depend on $v$, and $T(v,u)$ does not depend on $u$, cf. Figure \ref{tetrah}.

\begin{figure}[H]
	\centering
	\begin{tikzpicture}[scale=.55]
	\tikzstyle{nod1}= [circle, inner sep=0pt, fill=white, minimum size=4pt, draw]
	\tikzstyle{nod}= [circle, inner sep=0pt, fill=black, minimum size=4pt, draw]
	\def\lx{3}%
	\def\ly{1.22}%
	\def\l{4}%
	\def\d{4}%
\node[nod] (u00) at (0,0) [label=below: $u$] {};
	\node[nod1] (u10) at (\l,0) [label=below: $\widetilde{u}$] {};
	\node[nod1] (u01) at  (\lx,\ly) [label=below: $\widehat{u}$] {};
	\node[nod] (u11) at (\l+\lx,\ly) [label=below: $\widehat{\widetilde{u}}$] {};
	\node[nod1] (v00) at (0,\d) [label=above: {\color{red}$v$}] {};
	\node[nod] (v10) at (\l,\d) [label=below right: {\color{red}$\widetilde{v}$}] {};
	\node[nod] (v01) at (\lx,\d+\ly) [label=above: {\color{red}$\widehat{v}$}] {};
	\node[nod1] (v11) at (\l+\lx,\d+\ly) [label=above: {\color{red}$\widehat{\widetilde{v}}$}] {};

	\draw[thick]  (u00) -- (u10) -- (u11) -- (v11) -- (v10);
	\draw[thick]  (u00) -- (v00) -- (v01) -- (v11);
	\draw[thick]  (v00) -- (v10) -- (u10);
	\draw[thick,dashed]  (u00)  -- (u01) -- (u11);
	\draw[thick,dashed]  (u01) -- (v01);
	\end{tikzpicture}
\hspace{2cm}
	\begin{tikzpicture}[scale=.55]
	\tikzstyle{nod1}= [circle, inner sep=0pt, fill=white, minimum size=4pt, draw]
	\tikzstyle{nod}= [circle, inner sep=0pt, fill=black, minimum size=4pt, draw]
	\def\lx{3}%
	\def\ly{1.22}%
	\def\l{4}%
	\def\d{4}%
\node[nod1] (u00) at (0,0) [label=below: $u$] {};
	\node[nod] (u10) at (\l,0) [label=below: $\widetilde{u}$] {};
	\node[nod] (u01) at  (\lx,\ly) [label=below: $\widehat{u}$] {};
	\node[nod1] (u11) at (\l+\lx,\ly) [label=below: $\widehat{\widetilde{u}}$] {};
	\node[nod] (v00) at (0,\d) [label=above: {\color{red}$v$}] {};
	\node[nod1] (v10) at (\l,\d) [label=below right: {\color{red}$\widetilde{v}$}] {};
	\node[nod1] (v01) at (\lx,\d+\ly) [label=above: {\color{red}$\widehat{v}$}] {};
	\node[nod] (v11) at (\l+\lx,\d+\ly) [label=above: {\color{red}$\widehat{\widetilde{v}}$}] {};
	\draw[thick]  (u00) -- (u10) -- (u11) -- (v11) -- (v10);
	\draw[thick]  (u00) -- (v00) -- (v01) -- (v11);
	\draw[thick]  (v00) -- (v10) -- (u10);
	\draw[thick,dashed]  (u00)  -- (u01) -- (u11);
	\draw[thick,dashed]  (u01) -- (v01);
	\end{tikzpicture}
\caption{\label{tetrah} The stencils for a set of tetrahedron equations $T(u,v)=0$ (left) and $T(v,u)=0$ (right) are indicated by the black dots.  Each equation depends on $u$-variables (black) and $v$-variables (red).}
\end{figure}
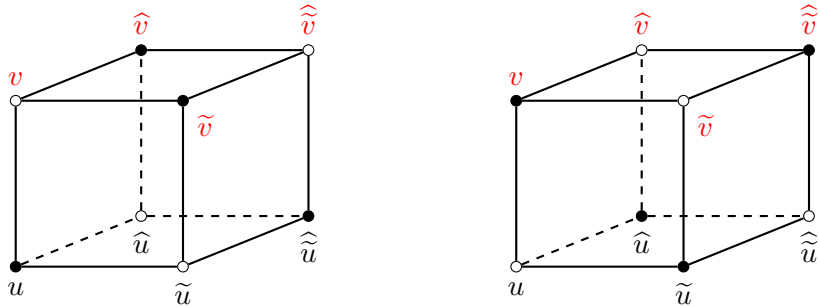

Importantly, both systems of equations $A=B=T=\h A=\t B=\b{T}=0$ and
$A=B=R=\h A=\t B=\b{R}=0$, although they are not defined on the faces of the cube, are CAC (irrespective of planarity, or the tetrahedron property). Indeed, in each case, from initial values $u,\t u,\h u, v$, there are three ways of calculating $\th v$, which yield the same value. 

\medskip
\noindent
{\bf Additive example.}
The natural auto-BT for H1 is given by
\[
A=(u-\t v)(\t u-v)-p+r=0,\qquad B=(u-\h v)(\h u- v)-q+r=0.
\]
We solve for $\t v,\h v$,
\begin{equation} \label{hvtv}
\h v = u+\frac{q - r}{v-\h u},\qquad  \t v = u +\frac{p - r}{v-\t u},
\end{equation}
and substitute the result into
\[
\h A=(\h u-\th v)(\th u- \h v)-p+r=0,\qquad \t B=(\t u- \th v)(\th u- \t v)-q+r=0.
\]
The obtained system for $\th u,\th v$ has two solutions:
\begin{itemize}
\item The first solution is
\begin{equation} \label{fs}
\th u = u + \frac{p - q}{\h u-\t u},\qquad
\th v = \t u + \frac{(q-r)(v-\t u)(\h u-\t u)}{p(\h u-v)+q(v-\t u)+r(\t u-\h u)},
\end{equation}
which, together with (\ref{hvtv}), satisfies $\HE([u,\h u],[\t u, \th u]; p, q) = \HE([v,\h v],[\t v, \th v]; p, q)=0$. The equation for $\th v$ in (\ref{fs}) has the tetrahedron property, as it does not depend on $u$. The equation for $\th u$ can be rewritten as a tetrahedron equation, by substituting the solution of $A=B=0$ with respect to $\t u, \h u$, i.e.
$\t u = v + (p -r)/(u - \t v)$, $\h u = v+ (q -r)/(u - \h v)$.
This gives
\begin{equation} \label{fsut}
\th u=\t v-\frac{(q - r)(u - \t v)(\h v - \t v)}{p(u - \h v) - q(u - \t v) + r(\h v - \t v)},
\end{equation}
which equals the equation for $\th v$ after interchanging $u\leftrightarrow v$.

\item The second solution is
\begin{equation}\label{seso}
\th u = u + \frac{p - r}{v-\t u} + \frac{q - r}{v-\h u},\qquad
\th v = \h u + \t u - v,
\end{equation}
which does not decouple. Substituting $\t u = v+(p - r)/(u - \t v)$ into the second equation yields a planar equation,
\begin{equation} \label{Ruv}
R(u,v)=(u - \t v )(\h u - \th v) + p - r = 0,
\end{equation}
and substituting $\h u = u + (q - r)/(u - \h v)$ into the first equation gives $R(v,u)=0$. Note that in (\ref{seso}) the equation for $\th v$ also has the tetrahedron property.
\end{itemize}

\subsection{Torqued auto-BTs}
We introduce an involution $\sigma$ which switches $u\leftrightarrow v$, that is $\sigma[u,v]=[v,u]$, and which commutes with shifts. We define an action of $\sigma$ on functions of 2-vectors, such as $A([u,v],[\t u,\t v])$ by $\sigma A([u,v],[\t u,\t v])=A(\sigma[u,v],\sigma[\t u,\t v])$. For functions of $u,v$ and their shifts, such as \eqref{QQ} or \eqref{RR}, we write $Q(u,v)=Q([u,v],[\t u,\t v],[\h u, \h v], [\th u,\th v])$ and then we simply have
\begin{align*}
\sigma Q(u,v)&=\sigma Q([u,v],[\t u,\t v],[\h u, \h v], [\th u,\th v])\\
&=Q(\sigma[u,v],\sigma[\t u,\t v],\sigma[\h u, \h v], \sigma[\th u,\th v])\\
&=Q([v,u],[\t v,\t u],[\h v, \h u],[\th v,\th u])\\
&=Q(v,u).
\end{align*}
We now formulate a useful lemma, which enables us to derive auto-BTs from auto-BTs.
\begin{lemma} \label{THM}
Suppose the system
\begin{align}
A([u,v],[\t u,\t v])=0,\qquad B([u,v],[\h u,\h v])=0,\label{BT1}
\end{align}
is invariant under $\sigma$ and that it is CAC with $\h A=\t B=0$ and $Q=\overline{Q}=0$, where
\begin{equation}\label{GQ}
Q=Q(u,v)=Q([u,v],[\t u,\t v],[\h u, \h v], [\th u,\th v])=0,
\end{equation}
and $\overline{Q}(u,v)=Q(v,u)$.
Then, with $a,b\in\{0,1\}$, the system
\begin{equation}
A_a=A([u,v],\sigma^a[\t u,\t v])=0, \qquad B_b=B([u,v],\sigma^b[\h u,\h v])=0, \label{gbt}
\end{equation}
is CAC with $\h A_a=\t B_b=0$ and $Q_{a,b}=\overline Q_{a,b}=0$, where
\begin{equation}
Q_{a,b}=Q([u,v],\sigma^a[\t u,\t v],\sigma^b[\h u, \h v],\sigma^{a+b}[\th u,\th v])=0,
\end{equation}
and $\overline{Q}_{a,b}(u,v)=Q_{a,b}(v,u)$.
\end{lemma}
\begin{proof}
The equations $\h A_a=\t B_b=0$ are equivalent to
\[
A(\sigma^b[\h u,\h v],\sigma^{a+b}[\th u,\th v])=0, \qquad B(\sigma^a[\t u,\t v],\sigma^{a+b}[\th u,\th v])=0
\]
because $A$ and $B$ are invariant under $\sigma$.
The consistency of the corresponding cube system then follows by relabeling of variables, cf. the proof of Lemma 1 in \cite{ZVZ}.
\end{proof}

The usefulness of this lemma lies in the following corollary. If the equation $Q_{a,b}=0$ depends on $u$-variables only, then the system $A_a=B_b=0$ provides an auto-BT for it. We distinguish two cases.
\begin{itemize}
\item If the system $Q=\overline{Q}=0$ is planar, we obtain the {\bf asymmetric torqued auto-BT}, (\ref{gbt}) with $a=1$, $b=0$,
\begin{equation}\label{abt}
A([u,v],[\t v,\t u])=0, \qquad B([u,v],[\h u,\h v])=0
\end{equation}
which is an auto-BT for
\begin{equation} \label{ASYC}
Q([u,v],[\t v,\t u],[\h u, \h v],[\th v,\th u])=0,
\end{equation}
which depends on $u$-variables only.
\item If the system $Q=\overline{Q}=0$ has the tetrahedron property, we have the {\bf symmetric torqued auto-BT}, (\ref{gbt}) with $a=b=1$,
\begin{equation}\label{sbt}
A([u,v],[\t v,\t u])=0, \qquad B([u,v],[\h v,\h u])=0,
\end{equation}
which is an auto-BT for
\begin{equation} \label{SYC}
Q([u,v],[\t v,\t u],[\h v, \h u],[\th u,\th v])=0,
\end{equation}
which depends on $u$-variables only.
\end{itemize}

\subsubsection{The asymmetric torqued auto-BT}
If the second solution to a BT, equation \eqref{RR}, is planar, then by relabeling $\t u\leftrightarrow \t v$ and $\th u\leftrightarrow \th v$, the equations decouple. After relabeling, the planar equation $R(u,v)=0$ will depend on $u$ only, and the planar equation $R(v,u)=0$ will depend on $v$ only, cf. Figure \ref{asymcube}.

\begin{figure}[H]
	\centering
	\begin{tikzpicture}[scale=.55]
	\tikzstyle{nod1}= [circle, inner sep=0pt, fill=white, minimum size=4pt, draw]
	\tikzstyle{nod}= [circle, inner sep=0pt, fill=black, minimum size=4pt, draw]
	\def\lx{3}%
	\def\ly{1.22}%
	\def\l{4}%
	\def\d{4}%
\node[nod] (u00) at (0,0) [label=below: $u$] {};
	\node[nod1] (u10) at (\l,0) [label=below: $\widetilde{u}$] {};
	\node[nod] (u01) at  (\lx,\ly) [label=below: $\widehat{u}$] {};
	\node[nod1] (u11) at (\l+\lx,\ly) [label=below: $\widehat{\widetilde{u}}$] {};
	\node[nod1] (v00) at (0,\d) [label=above: {\color{red}$v$}] {};
	\node[nod] (v10) at (\l,\d) [label=below right: {\color{red}$\widetilde{v}$}] {};
	\node[nod1] (v01) at (\lx,\d+\ly) [label=above: {\color{red}$\widehat{v}$}] {};
	\node[nod] (v11) at (\l+\lx,\d+\ly) [label=above: {\color{red}$\widehat{\widetilde{v}}$}] {};

	\draw[thick]  (u00) -- (u10) -- (u11) -- (v11) -- (v10);
	\draw[thick]  (u00) -- (v00) -- (v01) -- (v11);
	\draw[thick]  (v00) -- (v10) -- (u10);
	\draw[thick,dashed]  (u00)  -- (u01) -- (u11);
	\draw[thick,dashed]  (u01) -- (v01);
	\end{tikzpicture}
\hspace{2cm}
	\begin{tikzpicture}[scale=.55]
	\tikzstyle{nod1}= [circle, inner sep=0pt, fill=white, minimum size=4pt, draw]
	\tikzstyle{nod}= [circle, inner sep=0pt, fill=black, minimum size=4pt, draw]
	\def\lx{3}%
	\def\ly{1.22}%
	\def\l{4}%
	\def\d{4}%
\node[nod] (u00) at (0,0) [label=below: $u$] {};
	\node[nod1] (u10) at (\l,0) [label=below: {\color{red}$\widetilde{v}$}] {};
	\node[nod] (u01) at  (\lx,\ly) [label=below: $\widehat{u}$] {};
	\node[nod1] (u11) at (\l+\lx,\ly) [label=below: {\color{red}$\widehat{\widetilde{v}}$}] {};
	\node[nod1] (v00) at (0,\d) [label=above: {\color{red}$v$}] {};
	\node[nod] (v10) at (\l,\d) [label=below right: $\widetilde{u}$] {};
	\node[nod1] (v01) at (\lx,\d+\ly) [label=above: {\color{red}$\widehat{v}$}] {};
	\node[nod] (v11) at (\l+\lx,\d+\ly) [label=above: $\widehat{\widetilde{u}}$] {};
	\draw[thick]  (u00) -- (u10) -- (u11) -- (v11) -- (v10);
	\draw[thick]  (u00) -- (v00) -- (v01) -- (v11);
	\draw[thick]  (v00) -- (v10) -- (u10);
	\draw[thick,dashed]  (u00)  -- (u01) -- (u11);
	\draw[thick,dashed]  (u01) -- (v01);
	\end{tikzpicture}
\caption{ \label{asymcube} By relabeling of variables the diagonal planar stencil for $R(u,v)$ (left) becomes a quadrilateral for a decoupled equation $R(u)=0$ which depends on $u$-variables only (right). Similarly the expression $R(v,u)$ involves $v$-variables only (see the right cube in Figure \ref{planar}).}
\end{figure}
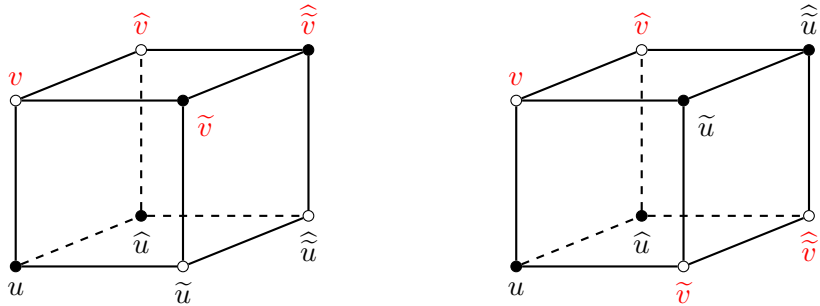

\begin{remark} The cubes on the right in Figures \ref{asymcube} and \ref{symcube} are useful to define equations. However, one should be aware that the fields $u,v$ are defined in the usual way, e.g. $\t{u}(n,m) = u(n + 1,m)$, cf. \cite[Remark 2.4]{ZVZ}.
\end{remark}

\noindent
{\bf Additive example, asymmetric case.} For the H1 equation, the asymmetric auto-BT (\ref{abt}) reads
\begin{equation}
A=(u-\t u)(\t v-v)-p+r=0,\qquad
B=(u-\h v)(\h u- v)-q+r=0. \label{abth1a}
\end{equation}
Switching $\t u\leftrightarrow \t v$ and $\th u\leftrightarrow \th v$ in \eqref{RR} with \eqref{Ruv}, gives
\[
(u - \t u )(\h u - \th u) + p - r = 0,
\qquad
(v - \t v )(\h v - \th v) + p - r = 0,
\]
which takes the form
\begin{equation}\label{abth1c}
\HE([u,\h u],\sigma[\t u,\th u];p,r)=0, \qquad \HE([v,\h v],\sigma[\t v,\th v];p,r)=0.
\end{equation}
These are torqued versions of H1, similar to the $A$-part of the auto-BT,
except that the dependence on the parameters is different. The $B$-part of the auto-BT is the standard H1 equation.

In order for (\ref{abth1a}) to be a proper auto-BT,
the equations (\ref{abth1c}) should not depend on the B\"acklund parameter $r$.
We achieve this by introducing a new parameter
\begin{equation} \label{pp}
p^\prime=p-r.
\end{equation}
We can write the $A$-part of the auto-BT as
\[
\HE([u,v],[\t v,\t u];p^\prime+r,r)=0,
\]
in terms of the new parameter. We note that for $\HE$ this may seem a bit odd as the equation does not depend on $r$.

\medskip
\noindent
{\bf Notation.}
We will adopt the notation $Q^a=0$ with
\begin{equation}\label{Qa}
Q^a(p,q)[u]=Q^a([u,\h u],[\t u,\th u]; p,q)=Q([u,\h u],\sigma[\t u,\th u]; p+q, q)
\end{equation}
for the torqued version of a lattice equation $Q=0$, with an $a$dditive transformation of the lattice parameters.

\medskip
\noindent
{\bf Additive example, asymmetric case, recap.}
Omitting the dependence on the fields, the equations (\ref{abth1c}) can both be written as
\begin{equation} \label{aH1}
\HE^a(p^\prime,q)=0,
\end{equation}
and the auto-BT, (\ref{abth1a}), is written as\footnote{Here $Q^a(p,r)$ and $Q(p,r)$ mean
\[Q^a(p,r)[u]=Q^a([u,\b u],[\t u,\t{\b u}]; p,r),~~ Q(p,r)[u]=Q([u,\b u],[\t u,\t{\b u}]; p,r),\]
respectively,  where $\b u=v$. If replacing $(\t u,p)$ by $(\h u, q)$, one gets $Q^a(q,r)$ and $Q(q,r)$.
}
\begin{equation} \label{abtH1}
A=\HE^a(p^\prime,r)=0,\qquad  B=\HE(q,r)=0.
\end{equation}
One may wonder why $q$ appears in \eqref{aH1}.
In fact, the equation does not depend on $q$, cf. $\HE^a$ in the list provided in section 6.
Note that the other additive torqued ABS (t-ABS) equations do have $q$-dependence.

\subsubsection{The symmetric torqued auto-BT}
By relabeling $\t u \leftrightarrow \t v$ and $\h u\leftrightarrow \h v$ a coupled system (\ref{TT}), which has the tetrahedron property, decouples, see Figure \ref{symcube}.

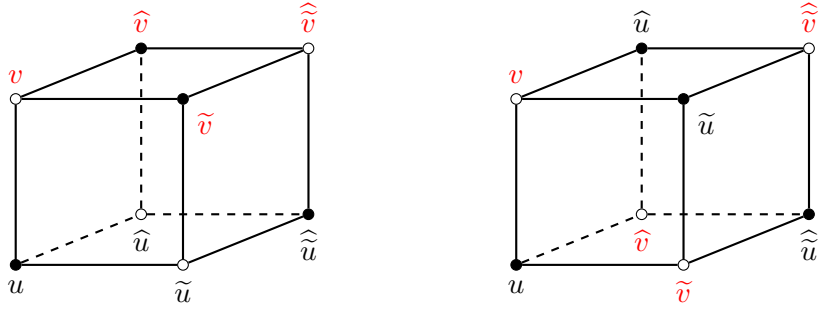
\begin{figure}[H]
	\centering
	\begin{tikzpicture}[scale=.55]
	\tikzstyle{nod1}= [circle, inner sep=0pt, fill=white, minimum size=4pt, draw]
	\tikzstyle{nod}= [circle, inner sep=0pt, fill=black, minimum size=4pt, draw]
	\def\lx{3}%
	\def\ly{1.22}%
	\def\l{4}%
	\def\d{4}%
\node[nod] (u00) at (0,0) [label=below: $u$] {};
	\node[nod1] (u10) at (\l,0) [label=below: $\widetilde{u}$] {};
	\node[nod1] (u01) at  (\lx,\ly) [label=below: $\widehat{u}$] {};
	\node[nod] (u11) at (\l+\lx,\ly) [label=below: $\widehat{\widetilde{u}}$] {};
	\node[nod1] (v00) at (0,\d) [label=above: {\color{red}$v$}] {};
	\node[nod] (v10) at (\l,\d) [label=below right: {\color{red}$\widetilde{v}$}] {};
	\node[nod] (v01) at (\lx,\d+\ly) [label=above: {\color{red}$\widehat{v}$}] {};
	\node[nod1] (v11) at (\l+\lx,\d+\ly) [label=above: {\color{red}$\widehat{\widetilde{v}}$}] {};

	\draw[thick]  (u00) -- (u10) -- (u11) -- (v11) -- (v10);
	\draw[thick]  (u00) -- (v00) -- (v01) -- (v11);
	\draw[thick]  (v00) -- (v10) -- (u10);
	\draw[thick,dashed]  (u00)  -- (u01) -- (u11);
	\draw[thick,dashed]  (u01) -- (v01);
	\end{tikzpicture}
\hspace{2cm}
	\begin{tikzpicture}[scale=.55]
	\tikzstyle{nod1}= [circle, inner sep=0pt, fill=white, minimum size=4pt, draw]
	\tikzstyle{nod}= [circle, inner sep=0pt, fill=black, minimum size=4pt, draw]
	\def\lx{3}%
	\def\ly{1.22}%
	\def\l{4}%
	\def\d{4}%
\node[nod] (u00) at (0,0) [label=below: $u$] {};
	\node[nod1] (u10) at (\l,0) [label=below: {\color{red}$\widetilde{v}$}] {};
	\node[nod1] (u01) at  (\lx,\ly) [label=below: {\color{red}$\widehat{v}$}] {};
	\node[nod] (u11) at (\l+\lx,\ly) [label=below: $\widehat{\widetilde{u}}$] {};
	\node[nod1] (v00) at (0,\d) [label=above: {\color{red}$v$}] {};
	\node[nod] (v10) at (\l,\d) [label=below right: $\widetilde{u}$] {};
	\node[nod] (v01) at (\lx,\d+\ly) [label=above: $\widehat{u}$] {};
	\node[nod1] (v11) at (\l+\lx,\d+\ly) [label=above: {\color{red}$\widehat{\widetilde{v}}$}] {};
	\draw[thick]  (u00) -- (u10) -- (u11) -- (v11) -- (v10);
	\draw[thick]  (u00) -- (v00) -- (v01) -- (v11);
	\draw[thick]  (v00) -- (v10) -- (u10);
	\draw[thick,dashed]  (u00)  -- (u01) -- (u11);
	\draw[thick,dashed]  (u01) -- (v01);
	\end{tikzpicture}
\caption{ \label{symcube} An equation $T(u,v)=0$ with the tetrahedron property (left). The black dots indicate the stencil on which the equation is defined. After relabeling variables the equation depends on $u$-variables only (right). Similarly the equation for $T(v,u)$ (see the right cube in Figure \ref{tetrah}) becomes an equation involving $v$-variables only.}
\end{figure}

\noindent
{\bf Additive example, symmetric case I.}
We consider the symmetric auto-BT (\ref{SYC}), which consists of
\begin{equation}
A=(u-\t u)(\t v-v)-p+r=0,\qquad B=(u-\h u)(\h v- v)-q+r=0. \label{sabh1a}
\end{equation}
We now switch $\t u \leftrightarrow \t v$ and $\h u\leftrightarrow \h v$ in the first solution, (\ref{fs}). This gives
\[
\th u = u + \frac{p - q}{\h v-\t v},\qquad
\th v = \t v + \frac{(q-r)(v-\t v)(\h v-\t v)}{p(\h v-v)+q(v-\t v)+r(\t v-\h v)}.
\]
The latter equation is decoupled, and the first one takes the same form
after elimination of $\t v,\h v$ making use of (\ref{sabh1a}). The equation can be identified as $\QE_0(p-r,q-r)=0$,
which is one of the ABS equations provided in the next section, again with different dependence on the parameters.
Using (\ref{pp}) and defining also
\begin{equation} \label{qp}
q^\prime=q-r,
\end{equation}
the equation $\QE_0(p^\prime,q^\prime)=0$ admits the auto-BT
\begin{equation} \label{sabth1}
\HE^a(p^\prime,r)=\HE^a(q^\prime,r)=0.
\end{equation}

\medskip
\noindent
{\bf Additive example, symmetric case II.}
We can also switch $\t u \leftrightarrow \t v$ and $\h u\leftrightarrow \h v$ in the second solution, (\ref{seso}).
This gives
\[
\th u = u + \frac{p - r}{v-\t v} + \frac{q - r}{v-\h v},\qquad
\th v = \h v + \t v - v,
\]
of which the first equation gives rise to the linear (difference) equation
\begin{equation} \label{eqD}
D: \th u +u - \t u - \h u = 0,
\end{equation} due to (\ref{sabh1a}). Thus, the auto-BT (\ref{sabth1}) is an auto-BT
for both $\QE_0(p^\prime,q^\prime)=0$ and $D=0$. It is quite special; for only two ABS equations ($\HE$ and $\HD_0$) we find decoupling in the second solution, applying the symmetric switch. It means that in these cases the symmetric auto-BT is weak. In \cite[Section 5.2-5.3]{HV} these extra solutions were called exotic.

\bigskip
\noindent
Note that in the asymmetric case both the equation (\ref{aH1}) and the auto-BT (\ref{abtH1}) only depend on $p^\prime$,
not on $p$. Similarly, for the symmetric case we have dependence on $p^\prime, q^\prime$,
and not on $p,q$. Therefore, in the sequel we will omit the prime.

\subsection{Pictorial representation, and superposition principles}\label{sec-3}

We will represent $Q^a(p,q)[u]=0$, equation \eqref{Qa}, pictorially as in Figure \ref{torq}.
The edge $[\t u,\th u]$ being torqued does not mean we switch $\t u $ and $\th u$ on the lattice,
but only in the equation.

\begin{figure}[H]
\centering
\begin{tikzpicture}[scale=.45]
	\tikzstyle{nod1}= [circle, inner sep=0pt, fill=white, minimum size=4pt, draw]
	\tikzstyle{nod}= [circle, inner sep=0pt, fill=black, minimum size=4pt, draw]
    \def\l{4}%
	\def\d{4}%
    \def\r{0.8}%

\fill[opacity=.2] (\r*\l,0) -- (\l,0) -- (\l,\d) -- (\r*\l,\d) -- (\r*\l,0);

\node[nod1] (u00) at (0,0) [label=below left: $u$] {};
	\node[nod1] (u10) at (\l,0) [label=below right: $\widetilde{u}$] {};
\node[nod1] (u01) at (0,\d) [label=above left: $\widehat{u}$] {};
	\node[nod1] (u11) at (\l,\d) [label=above right: $\widehat{\widetilde{u}}$] {};

\draw[font=\Large]  (\l/2,\d/2) node {$Q^a$};

\draw[thick]  (u00) -- (u10) -- (u11) -- (u01) -- (u00);

\end{tikzpicture}
\caption{\label{torq} Representation of a torqued equation.}
\end{figure}
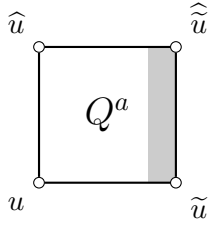

The results for the $\HE$ equation obtained in the previous subsection can be represented by the CAC systems in Figure \ref{ris}.

\begin{figure}[H]
	\centering
	\begin{tikzpicture}[scale=.55]
	\tikzstyle{nod1}= [circle, inner sep=0pt, fill=white, minimum size=4pt, draw]
	\tikzstyle{nod}= [circle, inner sep=0pt, fill=black, minimum size=4pt, draw]
	\def\lx{2.4}%
	\def\ly{1.22}%
	\def\l{4}%
	\def\d{4}%
    \def\r{0.8}%

\fill[opacity=.2] (\r*\l,0) -- (\l,0) -- (\l,\d) -- (\r*\l,\d) -- (\r*\l,0);
\fill[opacity=.2] (\r*\l+\lx ,\ly) -- (\l+\lx,\ly) -- (\l+\lx,\d+\ly) -- (\r*\l+\lx,\d+\ly) -- (\r*\l+\lx,\ly);
\fill[opacity=.2] (\r*\l,0) -- (\r*\l+\lx ,\ly) -- (\l+\lx,\ly) -- (\l,0) -- (\r*\l,0);
\fill[opacity=.2] (\r*\l,\d) -- (\r*\l+\lx ,\d+\ly) -- (\l+\lx,\d+\ly) -- (\l,\d) -- (\r*\l,\d);

\node[nod] (u00) at (0,0) [label=below: $u$] {};
	\node[nod] (u10) at (\l,0) [label=below: $\widetilde{u}$] {};
	\node[nod] (u01) at  (\lx,\ly) [label=below: $\widehat{u}$] {};
	\node[nod1] (u11) at (\l+\lx,\ly) [label=below: $\widehat{\widetilde{u}}$] {};
	\node[nod] (v00) at (0,\d) [label=above: $v$] {};
	\node[nod1] (v10) at (\l,\d) [label=below right: $\widetilde{v}$] {};
	\node[nod1] (v01) at (\lx,\d+\ly) [label=above: $\widehat{v}$] {};
	\node[nod1] (v11) at (\l+\lx,\d+\ly) [label=above: $\widehat{\widetilde{v}}$] {};

\draw[font=\Large]  (\l/2,\d/2) node {$\HE^a$};
\draw[font=\Large]  (\l+\lx/2,\d/2+\ly/2) node {$\HE$};
\draw[font=\Large]  (\l/2+\lx/2,\d+\ly/2) node {$\HE^a$};

	\draw[thick]  (u00) -- (u10) -- (u11) -- (v11) -- (v10);
	\draw[thick]  (u00) -- (v00) -- (v01) -- (v11);
	\draw[thick]  (v00) -- (v10) -- (u10);
	\draw[thick,dashed]  (u00)  -- (u01) -- (u11);
	\draw[thick,dashed]  (u01) -- (v01);

\draw (\l/2,-\d/4) node {(a)};
	\end{tikzpicture}
\hspace{1cm} \begin{tikzpicture}[scale=.55]
	\tikzstyle{nod1}= [circle, inner sep=0pt, fill=white, minimum size=4pt, draw]
	\tikzstyle{nod}= [circle, inner sep=0pt, fill=black, minimum size=4pt, draw]
	\def\lx{2.4}%
	\def\ly{1.22}%
	\def\l{4}%
	\def\d{4}%
    \def\r{0.8}%

\fill[opacity=.2] (\r*\l,0) -- (\l,0) -- (\l,\d) -- (\r*\l,\d) -- (\r*\l,0);
\fill[opacity=.2] (\r*\l+\lx ,\ly) -- (\l+\lx,\ly) -- (\l+\lx,\d+\ly) -- (\r*\l+\lx,\d+\ly) -- (\r*\l+\lx,\ly);
\fill[opacity=.2] (\r*\lx ,\r*\ly) -- (\lx,\ly) -- (\lx,\ly+\d) -- (\r*\lx,\r*\ly+\d) -- (\r*\lx,\r*\ly);
\fill[opacity=.2] (\l+\r*\lx ,\r*\ly) -- (\l+\lx,\ly) -- (\l+\lx,\ly+\d) -- (\l+\r*\lx,\r*\ly+\d) -- (\l+\r*\lx,\r*\ly);

\node[nod] (u00) at (0,0) [label=below: $u$] {};
	\node[nod] (u10) at (\l,0) [label=below: $\widetilde{u}$] {};
	\node[nod] (u01) at  (\lx,\ly) [label=below: $\widehat{u}$] {};
	\node[nod1] (u11) at (\l+\lx,\ly) [label=below: $\widehat{\widetilde{u}}$] {};
	\node[nod] (v00) at (0,\d) [label=above: $v$] {};
	\node[nod1] (v10) at (\l,\d) [label=below right: $\widetilde{v}$] {};
	\node[nod1] (v01) at (\lx,\d+\ly) [label=above: $\widehat{v}$] {};
	\node[nod1] (v11) at (\l+\lx,\d+\ly) [label=above: $\widehat{\widetilde{v}}$] {};

\draw[font=\Large]  (\l/2,\d/2) node {$\HE^a$};
\draw[font=\Large]  (\l+\lx/2,\d/2+\ly/2) node {$\HE^a$};
\draw[font=\Large]  (\l/2+\lx/2,\d+\ly/2) node {$\QE_0$};

	\draw[thick]  (u00) -- (u10) -- (u11) -- (v11) -- (v10);
	\draw[thick]  (u00) -- (v00) -- (v01) -- (v11);
	\draw[thick]  (v00) -- (v10) -- (u10);
	\draw[thick,dashed]  (u00)  -- (u01) -- (u11);
	\draw[thick,dashed]  (u01) -- (v01);
\draw (\l/2,-\d/4) node {(b)};
	\end{tikzpicture}
\hspace{1cm} \begin{tikzpicture}[scale=.55]
	\tikzstyle{nod1}= [circle, inner sep=0pt, fill=white, minimum size=4pt, draw]
	\tikzstyle{nod}= [circle, inner sep=0pt, fill=black, minimum size=4pt, draw]
	\def\lx{2.4}%
	\def\ly{1.22}%
	\def\l{4}%
	\def\d{4}%
    \def\r{0.8}%

\fill[opacity=.2] (\r*\l,0) -- (\l,0) -- (\l,\d) -- (\r*\l,\d) -- (\r*\l,0);
\fill[opacity=.2] (\r*\l+\lx ,\ly) -- (\l+\lx,\ly) -- (\l+\lx,\d+\ly) -- (\r*\l+\lx,\d+\ly) -- (\r*\l+\lx,\ly);
\fill[opacity=.2] (\r*\lx ,\r*\ly) -- (\lx,\ly) -- (\lx,\ly+\d) -- (\r*\lx,\r*\ly+\d) -- (\r*\lx,\r*\ly);
\fill[opacity=.2] (\l+\r*\lx ,\r*\ly) -- (\l+\lx,\ly) -- (\l+\lx,\ly+\d) -- (\l+\r*\lx,\r*\ly+\d) -- (\l+\r*\lx,\r*\ly);

\node[nod] (u00) at (0,0) [label=below: $u$] {};
	\node[nod] (u10) at (\l,0) [label=below: $\widetilde{u}$] {};
	\node[nod] (u01) at  (\lx,\ly) [label=below: $\widehat{u}$] {};
	\node[nod1] (u11) at (\l+\lx,\ly) [label=below: $\widehat{\widetilde{u}}$] {};
	\node[nod] (v00) at (0,\d) [label=above: $v$] {};
	\node[nod1] (v10) at (\l,\d) [label=below right: $\widetilde{v}$] {};
	\node[nod1] (v01) at (\lx,\d+\ly) [label=above: $\widehat{v}$] {};
	\node[nod1] (v11) at (\l+\lx,\d+\ly) [label=above: $\widehat{\widetilde{v}}$] {};

\draw[font=\Large]  (\l/2,\d/2) node {$\HE^a$};
\draw[font=\Large]  (\l+\lx/2,\d/2+\ly/2) node {$\HE^a$};
\draw[font=\Large]  (\l/2+\lx/2,\d+\ly/2) node {$D$};

	\draw[thick]  (u00) -- (u10) -- (u11) -- (v11) -- (v10);
	\draw[thick]  (u00) -- (v00) -- (v01) -- (v11);
	\draw[thick]  (v00) -- (v10) -- (u10);
	\draw[thick,dashed]  (u00)  -- (u01) -- (u11);
	\draw[thick,dashed]  (u01) -- (v01);
\draw (\l/2,-\d/4) node {(c)};
	\end{tikzpicture}
\caption{Consistent cubes, with torqued $\HE$-equations. \label{ris}}
\end{figure}
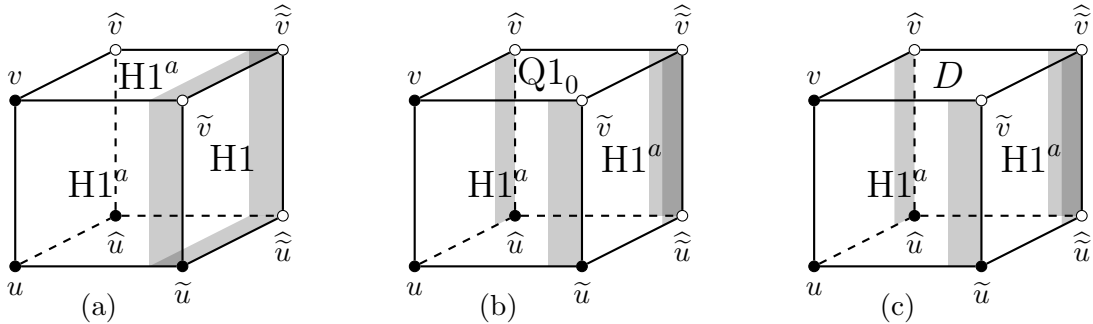

The consistent cubes in Figure \ref{ris} represent auto-BTs. Each auto-BT maps a solution, $u$,
to the equation on the bottom face, to a new solution (with parameter $r$), $v$, to the equation on the top face (in principle).
We now apply the same auto-BT but with parameter $s$ to both solutions $u,v$ to create two new solutions,
called $z,w$ respectively. A relation of the form $S(u,v,z,w,r,s)=0$ between the four solutions $u,v,z$
and $w$ is called a {\bf superposition principle}.

\begin{figure}[H]
\centering
\begin{tikzpicture}[scale=.5]
	\tikzstyle{nod1}= [circle, inner sep=0pt, fill=white, minimum size=4pt, draw]
	\tikzstyle{nod}= [circle, inner sep=0pt, fill=black, minimum size=4pt, draw]
	\def\lx{2.4}%
	\def\ly{1.22}%
	\def\l{4}%
	\def\d{4}%
    \def\ox{-2.5}
    \def\oy{-1.8}
    \def\s{1.9}
    \def\r{0.8}

\fill[opacity=.2] (\r*\l,0) -- (\l,0) -- (\l,\d) -- (\r*\l,\d) -- (\r*\l,0);
\fill[opacity=.2] (\r*\l+\lx ,\ly) -- (\l+\lx,\ly) -- (\l+\lx,\d+\ly) -- (\r*\l+\lx,\d+\ly) -- (\r*\l+\lx,\ly);
\fill[opacity=.2] (\r*\l,0) -- (\r*\l+\lx ,\ly) -- (\l+\lx,\ly) -- (\l,0) -- (\r*\l,0);
\fill[opacity=.2] (\r*\l,\d) -- (\r*\l+\lx ,\d+\ly) -- (\l+\lx,\d+\ly) -- (\l,\d) -- (\r*\l,\d);

\fill[opacity=.2] (\l*\s + \ox, \oy) -- (\l, 0) -- (\l*\r, 0) -- (\l*\r*\s + \ox, \oy) -- (\l*\s + \ox, \oy);
\fill[opacity=.2] (\l, \d) -- (\l*\s + \ox, \d*\s + \oy) -- (\l*\r*\s + \ox, \d*\s + \oy) -- (\l*\r, \d) -- (\l, \d);
\fill[opacity=.2] (\l*\s + \ox, \oy) -- (\l*\s + \ox, \d*\s + \oy) -- (\l*\r*\s + \ox, \d*\s + \oy) --
(\l*\r*\s + \ox, \oy) -- (\l*\s + \ox, \oy);

\fill[opacity=.2] (\l*\s + \ox, \d*\s + \oy) -- (\l*\s + \lx*\s + \ox, \d*\s + \ly*\s + \oy)
-- (\l*\r*\s + \lx*\s + \ox, \d*\s + \ly*\s + \oy) -- (\l*\r*\s + \ox, \d*\s + \oy) -- (\l*\s + \ox, \d*\s + \oy);
\fill[opacity=.2] (\lx + \l, \ly + \d) -- (\l*\s + \lx*\s + \ox, \d*\s + \ly*\s + \oy) -- (\l*\r*\s + \lx*\s + \ox, \d*\s
+ \ly*\s + \oy) -- (\l*\r + \lx, \ly + \d) -- (\lx + \l, \ly + \d);

\fill[opacity=.2] (\l*\s + \lx*\s + \ox, \ly*\s + \oy) -- (\l*\s + \lx*\s + \ox, \d*\s + \ly*\s + \oy)
-- (\l*\r*\s + \lx*\s + \ox, \d*\s + \ly*\s + \oy) -- (\l*\r*\s + \lx*\s + \ox, \ly*\s + \oy)
-- (\l*\s + \lx*\s + \ox, \ly*\s + \oy);
\fill[opacity=.2] (\lx + \l, \ly) -- (\l*\s + \lx*\s + \ox, \ly*\s + \oy) -- (\l*\r*\s + \lx*\s + \ox, \ly*\s + \oy)
-- (\l*\r + \lx, \ly) -- (\lx + \l, \ly);

\fill[opacity=.2] (\l*\s + \ox, \oy) -- (\l*\s + \lx*\s + \ox, \ly*\s + \oy) -- (\l*\r*\s + \lx*\s + \ox, \ly*\s + \oy)
-- (\l*\r*\s + \ox, \oy) -- (\l*\s + \ox, \oy);

\node[nod] (z00) at (\ox,\oy) [label=below: $z$] {};
    \node[nod1] (z10) at (\ox+\s*\l,\oy) [label=below: $\widetilde{z}$] {};
	\node[nod1] (z01) at  (\ox+\s*\lx,\oy+\s*\ly) [label=left: {\tiny $\widehat{z}$}] {};
	\node[nod1] (z11) at (\ox+\s*\l+\s*\lx,\oy+\s*\ly) [label=below: $\widehat{\widetilde{z}}$] {};
	\node[nod1] (w00) at (\ox,\oy+\s*\d) [label=above: $w$] {};
	\node[nod1] (w10) at (\ox+\s*\l,\oy+\s*\d) [label=above: $\widetilde{w}$] {};
	\node[nod1] (w01) at (\ox+\s*\lx,\oy+\s*\d+\s*\ly) [label=above: $\widehat{w}$] {};
	\node[nod1] (w11) at (\ox+\s*\l+\s*\lx,\oy+\s*\d+\s*\ly) [label=above: $\widehat{\widetilde{w}}$] {};

    \node[nod] (u00) at (0,0) [label=left: $u$] {};
	\node[nod] (u10) at (\l,0) [label=below: $\widetilde{u}$] {};
	\node[nod] (u01) at  (\lx,\ly) [label=above right: $\widehat{u}$] {};
	\node[nod1] (u11) at (\l+\lx,\ly) [label=above right: $\widehat{\widetilde{u}}$] {};
	\node[nod] (v00) at (0,\d) [label=below right: $v$] {};
	\node[nod1] (v10) at (\l,\d) [label=below right: $\widetilde{v}$] {};
	\node[nod1] (v01) at (\lx,\d+\ly) [label= below right: $\widehat{v}$] {};
	\node[nod1] (v11) at (\l+\lx,\d+\ly) [label=below right: $\widehat{\widetilde{v}}$] {};

\draw[font=\Large]  (\ox+\s*\l/2,\oy+\s*\d/2) node {$\HE^a$};
\draw[font=\Large]  (\ox+\s*\l+\s*\lx/2,\oy+\s*\d/2+\s*\ly/2) node {$\HE$};
\draw[font=\Large]  (\ox+\s*\l/2+\s*\lx/2,\oy+\s*\d+\s*\ly/2) node {$\HE^a$};

\draw[font=\Large]  (\ox/2,\oy/2+\s*\d/4+\d/4) node {$\HE$};

	\draw[thick]  (u00) -- (u10) -- (u11) -- (v11) -- (v10);
	\draw[thick]  (u00) -- (v00) -- (v01) -- (v11);
	\draw[thick]  (v00) -- (v10) -- (u10);
	\draw[thick,dashed]  (u00)  -- (u01) -- (u11);
	\draw[thick,dashed]  (u01) -- (v01);

	\draw[thick]  (z00) -- (z10) -- (z11) -- (w11) -- (w10);
	\draw[thick]  (z00) -- (w00) -- (w01) -- (w11);
	\draw[thick]  (w00) -- (w10) -- (z10);
	\draw[thick,dashed]  (z00)  -- (z01) -- (z11);
	\draw[thick,dashed]  (z01) -- (w01);

\draw (u00) -- (z00);
\draw (u10) -- (z10);
\draw (u01) -- (z01);
\draw (u11) -- (z11);
\draw (v00) -- (w00);
\draw (v10) -- (w10);
\draw (v01) -- (w01);
\draw (v11) -- (w11);

\draw (\l/2,-2*\d/3) node {(a)};
\end{tikzpicture}
\hspace{1cm}
\begin{tikzpicture}[scale=.5]
	\tikzstyle{nod1}= [circle, inner sep=0pt, fill=white, minimum size=4pt, draw]
	\tikzstyle{nod}= [circle, inner sep=0pt, fill=black, minimum size=4pt, draw]
	\def\lx{2.4}%
	\def\ly{1.22}%
	\def\l{4}%
	\def\d{4}%
    \def\ox{-2.5}
    \def\oy{-1.8}
    \def\s{1.9}
    \def\r{0.8}

\fill[opacity=.2] (\r*\l,0) -- (\l,0) -- (\l,\d) -- (\r*\l,\d) -- (\r*\l,0);
\fill[opacity=.2] (\r*\l+\lx ,\ly) -- (\l+\lx,\ly) -- (\l+\lx,\d+\ly) -- (\r*\l+\lx,\d+\ly) -- (\r*\l+\lx,\ly);
\fill[opacity=.2] (\r*\lx ,\r*\ly) -- (\lx,\ly) -- (\lx,\ly+\d) -- (\r*\lx,\r*\ly+\d) -- (\r*\lx,\r*\ly);
\fill[opacity=.2] (\l+\r*\lx ,\r*\ly) -- (\l+\lx,\ly) -- (\l+\lx,\ly+\d) -- (\l+\r*\lx,\r*\ly+\d) -- (\l+\r*\lx,\r*\ly);

\fill[opacity=.2] (\l*\s + \ox, \oy) -- (\l, 0) -- (\l*\r, 0) -- (\l*\r*\s + \ox, \oy) -- (\l*\s + \ox, \oy);
\fill[opacity=.2] (\l, \d) -- (\l*\s + \ox, \d*\s + \oy) -- (\l*\r*\s + \ox, \d*\s + \oy) -- (\l*\r, \d) -- (\l, \d);
\fill[opacity=.2] (\l*\s + \ox, \oy) -- (\l*\s + \ox, \d*\s + \oy) -- (\l*\r*\s + \ox, \d*\s + \oy) -- (\l*\r*\s + \ox, \oy)
-- (\l*\s + \ox, \oy);

\fill[opacity=.2] (\lx + \l, \ly + \d) -- (\l*\s + \lx*\s + \ox, \d*\s + \ly*\s + \oy)
-- (\l*\r*\s + \lx*\s + \ox, \d*\s + \ly*\s + \oy) -- (\l*\r + \lx, \ly + \d) -- (\lx + \l, \ly + \d);

\fill[opacity=.2] (\l*\s + \lx*\s + \ox, \ly*\s + \oy) -- (\l*\s + \lx*\s + \ox, \d*\s + \ly*\s + \oy)
-- (\l*\r*\s + \lx*\s + \ox, \d*\s + \ly*\s + \oy) -- (\l*\r*\s + \lx*\s + \ox, \ly*\s + \oy)
-- (\l*\s + \lx*\s + \ox, \ly*\s + \oy);
\fill[opacity=.2] (\lx + \l, \ly) -- (\l*\s + \lx*\s + \ox, \ly*\s + \oy) -- (\l*\r*\s + \lx*\s + \ox, \ly*\s + \oy)
-- (\l*\r + \lx, \ly) -- (\lx + \l, \ly);

\fill[opacity=.2] (\lx + \l, \ly) -- (\l*\s + \lx*\s + \ox, \ly*\s + \oy) -- (\lx*\r*\s + \l*\s + \ox, \ly*\r*\s + \oy)
-- (\lx*\r + \l, \ly*\r) -- (\lx + \l, \ly);
\fill[opacity=.2] (\l*\s + \lx*\s + \ox, \ly*\s + \oy) -- (\l*\s + \lx*\s + \ox, \d*\s + \ly*\s + \oy)
-- (\lx*\r*\s + \l*\s + \ox, \ly*\r*\s + \d*\s + \oy) -- (\lx*\r*\s + \l*\s + \ox, \ly*\r*\s + \oy)
-- (\l*\s + \lx*\s + \ox, \ly*\s + \oy);
\fill[opacity=.2] (\lx + \l, \ly + \d) -- (\l*\s + \lx*\s + \ox, \d*\s + \ly*\s + \oy)
-- (\lx*\r*\s + \l*\s + \ox, \ly*\r*\s + \d*\s + \oy) -- (\lx*\r + \l, \ly*\r + \d) -- (\lx + \l, \ly + \d);

\fill[opacity=.2] (\lx, \ly + \d) -- (\lx*\s + \ox, \d*\s + \ly*\s + \oy) -- (\lx*\r*\s + \ox, \ly*\r*\s + \d*\s + \oy)
-- (\lx*\r, \ly*\r + \d) -- (\lx, \ly + \d);
\fill[opacity=.2] (\lx*\s + \ox, \d*\s + \ly*\s + \oy) -- (\lx*\s + \ox, \ly*\s + \oy) -- (\lx*\r*\s + \ox, \ly*\r*\s + \oy)
-- (\lx*\r*\s + \ox, \ly*\r*\s + \d*\s + \oy) -- (\lx*\s + \ox, \d*\s + \ly*\s + \oy);
\fill[opacity=.2] (\lx, \ly) -- (\lx*\s + \ox, \ly*\s + \oy) -- (\lx*\r*\s + \ox, \ly*\r*\s + \oy)
-- (\lx*\r, \ly*\r) -- (\lx, \ly);

\node[nod] (z00) at (\ox,\oy) [label=below: $z$] {};
    \node[nod1] (z10) at (\ox+\s*\l,\oy) [label=below: $\widetilde{z}$] {};
	\node[nod1] (z01) at  (\ox+\s*\lx,\oy+\s*\ly) [label=left: {\tiny $\widehat{z}$}] {};
	\node[nod1] (z11) at (\ox+\s*\l+\s*\lx,\oy+\s*\ly) [label=below: $\widehat{\widetilde{z}}$] {};
	\node[nod1] (w00) at (\ox,\oy+\s*\d) [label=above: $w$] {};
	\node[nod1] (w10) at (\ox+\s*\l,\oy+\s*\d) [label=above: $\widetilde{w}$] {};
	\node[nod1] (w01) at (\ox+\s*\lx,\oy+\s*\d+\s*\ly) [label=above: $\widehat{w}$] {};
	\node[nod1] (w11) at (\ox+\s*\l+\s*\lx,\oy+\s*\d+\s*\ly) [label=above: $\widehat{\widetilde{w}}$] {};

    \node[nod] (u00) at (0,0) [label=left: $u$] {};
	\node[nod] (u10) at (\l,0) [label=below: $\widetilde{u}$] {};
	\node[nod] (u01) at  (\lx,\ly) [label=above right: $\widehat{u}$] {};
	\node[nod1] (u11) at (\l+\lx,\ly) [label=above right: $\widehat{\widetilde{u}}$] {};
	\node[nod] (v00) at (0,\d) [label=below right: $v$] {};
	\node[nod1] (v10) at (\l,\d) [label=below right: $\widetilde{v}$] {};
	\node[nod1] (v01) at (\lx,\d+\ly) [label= below right: $\widehat{v}$] {};
	\node[nod1] (v11) at (\l+\lx,\d+\ly) [label=below right: $\widehat{\widetilde{v}}$] {};

\draw[font=\Large]  (\ox+\s*\l/2,\oy+\s*\d/2) node {$\HE^a$};
\draw[font=\Large]  (\ox+\s*\l+\s*\lx/2,\oy+\s*\d/2+\s*\ly/2) node {$\HE^a$};
\draw[font=\Large]  (\ox+\s*\l/2+\s*\lx/2,\oy+\s*\d+\s*\ly/2) node {$\QE_0$};

\draw[font=\Large]  (\ox/2,\oy/2+\s*\d/4+\d/4) node {$\HE$};

	\draw[thick]  (u00) -- (u10) -- (u11) -- (v11) -- (v10);
	\draw[thick]  (u00) -- (v00) -- (v01) -- (v11);
	\draw[thick]  (v00) -- (v10) -- (u10);
	\draw[thick,dashed]  (u00)  -- (u01) -- (u11);
	\draw[thick,dashed]  (u01) -- (v01);

	\draw[thick]  (z00) -- (z10) -- (z11) -- (w11) -- (w10);
	\draw[thick]  (z00) -- (w00) -- (w01) -- (w11);
	\draw[thick]  (w00) -- (w10) -- (z10);
	\draw[thick,dashed]  (z00)  -- (z01) -- (z11);
	\draw[thick,dashed]  (z01) -- (w01);

\draw (u00) -- (z00);
\draw (u10) -- (z10);
\draw (u01) -- (z01);
\draw (u11) -- (z11);
\draw (v00) -- (w00);
\draw (v10) -- (w10);
\draw (v01) -- (w01);
\draw (v11) -- (w11);

\draw (\l/2,-2*\d/3) node {(b)};
\end{tikzpicture}

\caption{\label{4DC} 4D cubes representing auto-BTs with their superposition principles.}
\end{figure}
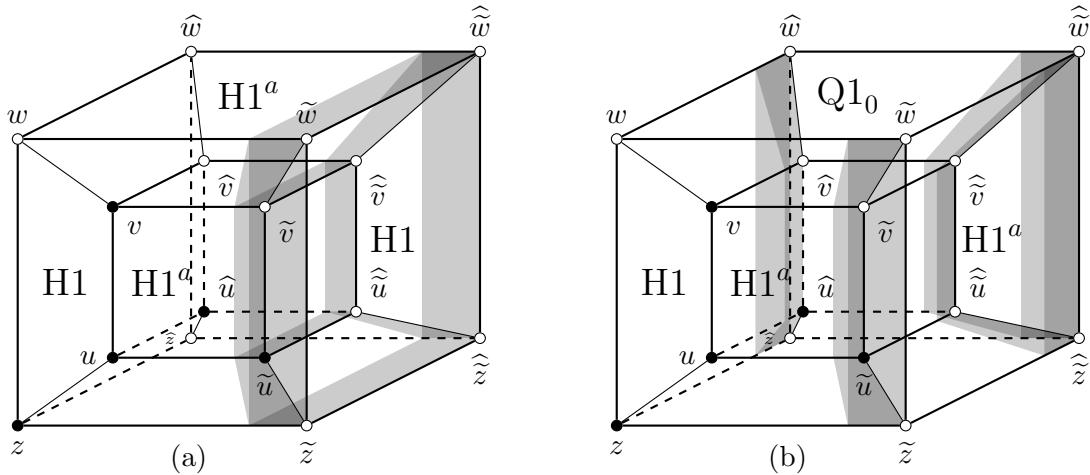

It can be seen from Figure \ref{4DC} that the superposition principles for both the asymmetric auto-BT
and the symmetric auto-BT is equal to the equation $\HE$.
Each 4D-cube consists of 7 3D-cubes (center, top, bottom, front, back, left and right),
excluding the outer 3D-cube. We have placed the 3D cube from Figure \ref{ris}(a)
in the center of Figure \ref{4DC}(a) and, as we apply the same auto-BT to both $u$ and $v$,
we place Figure \ref{ris}(a) also in the top and bottom cubes of Figure \ref{4DC}(a).
As parallel faces carry the same equation,
in the front and back cubes the same configuration as in Figure \ref{ris}(a) appears,
whereas in the left and right cubes we find no torqued equations.
All four cubes left, right, front and back are CAC if we impose the equation $\HE=0$ on the inner faces,
i.e. the face $z-u-v-w$ and their shifts. Similarly, placing the cube from Figure \ref{ris}(b) in the center,
top and bottom cube of a 4D-cube yields the cube from Figure \ref{ris}(a) on the remaining 4 cubes,
as depicted in Figure \ref{4DC}(b), where the equation $\QE_0$ may be replaced by $D$ (\ref{eqD}).

\section{Torqued auto-BTs for ABS equations and torqued ABS equations}\label{sec-4}
The ABS equations are
\begin{align*}
\HE:\quad &(u-\h{\t{u}})(\t{u}-\h{u})-p+q=0,\\
\HT:\quad &(u-\h{\t{u}})(\t{u}-\h{u})-(p-q)(u+\t{u}+\h{u}+\h{\t{u}}+p+q)=0,\\
\HD_\delta:\quad &p(u\t{u}+\h{u}\h{\t{u}})-q(u\h{u}+\t{u}\h{\t{u}})+\delta(p^2-q^2)=0,\\
\AE_\delta:\quad &p(u+\h{u})(\t{u}+\h{\t{u}})-q(u+\t{u})(\h{u}+\h{\t{u}})-\delta^2pq(p-q)=0,\\
\AT:\quad &p(1-q^2)(u\t{u}+\h{u}\h{\t{u}})-q(1-p^2)(u\h{u}+\t{u}\h{\t{u}})
-(p^2-q^2)(1+u\t{u}\h{u}\h{\t{u}})=0,\\
\QE_\delta:\quad &p(u-\h{u})(\t{u}-\h{\t{u}})-q(u-\t{u})(\h{u}-\h{\t{u}})+\delta^2pq(p-q)=0,\\
\QT: \quad &p(u-\h{u})(\t{u}-\h{\t{u}})-q(u-\t{u})(\h{u}-\h{\t{u}})
+pq(p-q)(u+\t{u}+\h{u}+\h{\t{u}}-p^2+pq-q^2)=0,\\
\QD_\delta:\quad &p(1-q^2)(u\h{u}+\t{u}\h{\t{u}})-q(1-p^2)(u\t{u}+\h{u}\h{\t{u}})
-(p^2-q^2)\left(\t{u}\h{u}+u\h{\t{u}}+\frac{\delta^2(1-p^2)(1-q^2)}{4pq}\right)=0,\\
\QV:\quad &\mathrm{sn}(p)(u\t{u}+\h{u}\h{\t{u}})-\mathrm{sn}(q)(u\h{u}+\t{u}\h{\t{u}})
+\mathrm{sn}(p-q)\Big(k\,\mathrm{sn}(p)\mathrm{sn}(q)(u\t{u}\h{u}\h{\t{u}}+1)
-\t{u}\h{u}-u\h{\t{u}} \Big)=0,
\end{align*}
where we have taken the Hietarinta form of $\QV$, found in \cite{Hie},
in which $k$ is the elliptic modulus of the Jacobi sine function $\mathrm{sn}$.

For some of the ABS equations the parameter transformation in their torqued counterpart is different than the additive one we have seen for $\HE$. We provide details for $\HD_0$, where the parameter transformation is multiplicative.

\medskip
\noindent
{\bf Multiplicative example, $\HD_0$.} The natural auto-BT for $\HD_0$ is given by
\begin{equation} \label{H3BT}
A=p(u\t{u}+v\t{v})-r(uv+\t{u}\t{v})=0,\qquad B=q(u\h{u}+v\h{v})-r(uv+\h{u}\h{v})=0.
\end{equation}
Solving the equations $A=\widehat{A}=B=\widetilde{B}=0$ yields two solutions, which gives rise to:
\begin{itemize}
\item a decoupled system,
\begin{equation*}
    p(u\t{u}+\h{u}\th{u})-q(u\h{u}+\t{u}\th{u})=0,\qquad p(v\t{v}+\h{v}\th{v})-q(v\h{v}+\t{v}\th{v})=0,
\end{equation*}
which is $\HD_0[u]=\HD_0[v]=0$, and can be written in tetrahedral form \eqref{TT} as
\begin{equation} \label{dH3}
\begin{split}
&p(r^2-q^2)( u\h v + \th u\t v) + q(p^2-r^2)(\th u\h v + u\t v) + r(q^2-p^2)(u\th u + \t v \h v)=0,\\
&p(r^2-q^2)(\h u v + \t u \th v) + q(p^2-r^2)(\h u\th v + \t u v) + r(q^2-p^2)(\t u \h u + v\th v)=0,
\end{split}
\end{equation}
\item and a coupled system,
\begin{equation} \label{cH3}
\t{u}\h{u}-v\th{v}=0,\qquad  \t{v}\h{v}-u\th{u}=0,
\end{equation}
which happens to be tetrahedral, and can be written in planar form \eqref{RR} as
\begin{equation} \label{pcH3}
(u\th v + \h u \t v)p - (u\h u + \t v \th v)r=0,\qquad (v\th u + \h v \t u)p - (v\h v + \t u \th u)r=0.
\end{equation}
\end{itemize}

\noindent
{\bf Asymmetric auto-BT.}
The asymmetric auto-BT is obtained by switching $\t u\leftrightarrow \t v$ in (\ref{H3BT}),
\begin{equation} \label{aH3BT}
A=p(u\t{v}+v\t{u})-r(uv+\t{u}\t{v})=0,\qquad B=q(u\h{u}+v\h{v})-r(uv+\h{u}\h{v})=0.
\end{equation}
Switching $\t u\leftrightarrow \t v$ and $\th u\leftrightarrow \th v$ in the coupled system (\ref{pcH3}) leads to the decoupled system
\begin{equation} \label{tH3}
p(u \th u + \t u \h u) - r(u\h u + \t u \th u)=0,\qquad
p(v \th v + \t v \h v) - r(v\h v + \t v \th v)=0.
\end{equation}
These equations should not depend on $r$. So we introduce the variable
\begin{equation} \label{ps}
p^\ast=p/r
\end{equation}
and the following notation.

\medskip
\noindent
{\bf Notation.}
We will adopt the notation $Q^m=0$ with
\begin{equation} \label{mnot}
Q^m(p,q)[u]=Q^m([u,\h u],[\t u,\th u]; p,q):=Q([u,\h u],\sigma[\t u,\th u]; pq, q)
\end{equation}
for the multiplicative torqued version of a lattice equation $Q=0$.

\medskip
\noindent
By doing so, the equations (\ref{tH3}) are captured by
\[
\HD_0^m(p^\ast,q)[u]=\HD_0^m(p^\ast,q)[v]=0,
\]
and the $A$-part of the auto-BT (\ref{aH3BT}) becomes
\[
A=\HD_0^m([u,v],[\t u,\t v];p^\ast,r)=0.
\]

\medskip
\noindent
{\bf Symmetric auto-BT.}
The symmetric auto-BT is obtained by switching $\t u\leftrightarrow \t v$ and $\h u\leftrightarrow \h v$ in (\ref{H3BT}),
\begin{equation} \label{sH3BT}
A=p(u\t{v}+v\t{u})-r(uv+\t{u}\t{v})=0,\qquad B=q(u\h{v}+v\h{u})-r(uv+\h{u}\h{v})=0,
\end{equation}
which is conveniently written as $\HD_0^m(p^\ast,r)=\HD_0^m(q^\ast,r)=0$,
where $q^\ast=q/r$.
Switching $\t u\leftrightarrow \t v$ and $\h u\leftrightarrow \h v$ in the tetrahedral system (\ref{dH3}), leads to the decoupled system (up to an irrelevant factor $r^3$)
\begin{equation}
\QD_0(p^\ast,q^\ast)[u]=\QD_0(p^\ast,q^\ast)[v]=0.
\end{equation}

As for the $\HE$ equation, switching $\t u\leftrightarrow \t v$ and $\h u\leftrightarrow \h v$
in the tetrahedral system (\ref{cH3}) leads to another decoupled system, of which the $u$-equation is the quotient equation
\begin{equation} \label{eqK}
K: u\th u - \t u\h u =0.
\end{equation}

\bigskip
\noindent
For equations $\HD_\delta,\AT$ and $\QD_\delta$ the dependence on the parameters is the same as in the $\HD_0$ case. As we no longer have dependence on $p,q$, we will omit the asterix in the sequel (remember, we were also going to omit the prime we used for the additive parameters). For each ABS equation the natural auto-BT is an auto-BT for an equation which depends on both $u,v$ and which is planar. As we also know each ABS equation has the tetrahedron property, this leads to the following theorem, whose proof is obtained by direct calculation.
\begin{theorem}\label{Thm}
For each ABS equation both system (\ref{abt}) and system (\ref{sbt}) provide an auto-BT for a quad-equation. For equations $\HT,\AE,\QE,\QT$ and $\QV$ the dependence on the parameters is additive. For equations $\HD,\AT$ and $\QD$ the dependence on the parameters is multiplicative.
The asymmetric and symmetric torqued auto-BTs (t-auto-BTs) and the equations they give rise to are provided in Table \ref{Tafel1} and Table \ref{Tafel2}.
\begin{table}[H]
\begin{center}
\begin{tabular}{l|l|l}
Equation & Asymmetric t-auto-BT & Symmetric t-auto-BT \\
$Q(p,q)=0$ & $Q^a(p,r)=Q(q,r)=0$ & $Q^a(p,r)=Q^a(q,r)=0$ \\
\hline
$\HE$ & $\HE^a$ & $\QE_0$ or $D$\\
$\HT$ & $\HT^a$ & $\QE_1$\\
$\AE_\delta$ & $\AE^a_\delta$ & $\QE_\delta$\\
$\QE_\delta$ & $\QE_\delta^a$ & $\QE_\delta$\\
$\QT$ & $\QT^a$ & $\QT$\\
$\QV$ & $\QV^a$ & $\QV$
\end{tabular}
\caption{\label{Tafel1} Torqued auto-B\"acklund transformations, additive cases.}
\end{center}
\end{table}
\begin{table}[H]
\begin{center}
\begin{tabular}{l|l|l}
\hline
Equation & Asymmetric t-auto-BT & Symmetric t-auto-BT \\
$Q(p,q)=0$ & $Q^m(p,r)=Q(q,r)=0$ & $Q^m(p,r)=Q^m(q,r)=0$ \\
\hline
$\HD_0$ & $\HD^m_0$ & $K$\\
$\HD_\delta$ & $\HD^m_\delta$ & $\QD_0$\\
$\AT$ & $\AT^m$ & $\QD_0$\\
$\QD_\delta$ & $\QD^m_\delta$ & $\QD_\delta$\\
\end{tabular}
\caption{\label{Tafel2} Torqued auto-B\"acklund transformations, multiplicative cases.}
\end{center}
\end{table}
For each auto-BT mentioned in Tables \ref{Tafel1} and \ref{Tafel2},
the superposition principle is given by the original equation, i.e. it takes the form
\begin{equation}
Q([u,z],[v,w];\ r,s)=0,
\end{equation}
cf. Figure \ref{4DC}.
\end{theorem}

Each additive torqued equation of type $Q$ satisfies $Q^a(p,q)=Q(-p,q)$ (for $\QV$ this is due to the anti-symmetry of the Jacobi sine function), whereas for the multiplicative equation $\QD_\delta^m(p,q)=qp^2\QD(1/p,q)$, and for the $A$-equations we have
$
\AE^a([u,\h u],[\t u, \th u];\ p,q)=\AE([u,\h u],[-\t u, -\th u];\ -p,q)$, and $\AT^a([u,\h u],[\t u, \th u];\ p,q)=qp^2\t u\th u\AT([u,\h u],[1/\t u,1/\th u];\ 1/p,q)$. It is not difficult to verify that we can adjust the signs on the whole cube, cf. \cite[proof of Theorem 4]{ABS09}, to show that, for these equations, the torqued cubes are equivalent to the natural ones.

The consistent cubes with H-type equations are special cases of consistent cubes listed in \cite{Boll,RB11}. This connection will be made precise in section 6. Whereas for $\epsilon=0$ the rhombic version of $H^\epsilon_i$ corresponds to ABS equation H$_i$, the trapezoidal version of $H^0_i$ corresponds to a t-ABS equation.

The B\"acklund transformations given in \cite[Table 2]{Atk08} consist of torqued ABS equations. Our result explains how the corresponding  consistent cubes relate to the CAC property of ABS-equations.

\section{Torqued non-natural auto-BTs for $\HD_\delta,\AT$ and $\QD_0$}
One does not have to take the natural auto-BT as the starting point. For example, consider the equations
\begin{equation} \label{dabth3}
A=u\t u+v\t v+\delta p=0,\qquad
B=u\h u+v\h v+\delta q=0,
\end{equation}
which provide an auto-BT for H$3_\delta$ \cite{ZZ}, and note that it does not depend on a B\"acklund parameter.
Performing the symmetric switch we obtain
\begin{equation} \label{dabth3s}
u\t v+v\t u+\delta p=0,\qquad
u\h v+v\h u+\delta q=0,
\end{equation}
which provide an auto-BT for H$3_0$. Here, the parameter $\delta$ now acts as a B\"acklund parameter. The asymmetric switch does not lead to a decoupled system.

To find a superposition principle for the auto-BT (\ref{dabth3}),
we pose the following equations on the front cube of a 4D cube (see Figure \ref{4DC} for the labeling of the vertices):
\begin{equation}\label{fc}
\begin{split}
u\t u+v\t v+\delta p=0 &\qquad (\text{back}) \\
z\t z+w\t w+\delta p=0 &\qquad (\text{front}) \\
u\t u+z\t z+\delta p=0 &\qquad (\text{bottom}) \\
v\t v+w\t w+\delta p=0 &\qquad (\text{top}).
\end{split}
\end{equation}
These equations are coupled and we can neither eliminate all variables $u,v,w,z$
nor their upshifts $\t u,\t v,\t w,\t z$. We can however derive the equations
\[
u\t u - w\t w =0,\qquad  v\t v - z\t z =0,
\]
which show that a superposition principle is given by the (reducible) equation
\begin{equation} \label{dspp}
(u-w)(v-z)=0.
\end{equation}
The equations (\ref{fc}) decouple after application of the symmetric switch, here we can derive both equations $u=w$ and $v=z$, as well as their upshifted  versions.
Consequently, (\ref{dspp}) is also a superposition principle for (\ref{dabth3s}).

\medskip
The non-natural auto-BT (\ref{dabth3}) can be conveniently written as
\begin{equation} \label{DaBT}
Q(p,0)=Q(q,0)=0,
\end{equation}
where $Q=\HD_\delta$. Both equations $\AT$ and $\QD_0$ give rise to non-natural auto-BTs of the same form (\ref{DaBT}), cf. \cite{ZZ}, and to torqued versions thereof, as per the below theorem.

\medskip
{\bf Notation.}
We will adopt the notation $Q^t=0$ with
\begin{equation}
Q^t(p,q)[u]=Q^t([u,\h u],[\t u,\th u]; p,q)=Q([u,\h u],\sigma[\t u,\th u]; p, q)
\end{equation}
for the (plain) torqued version of a lattice equation $Q=0$.

\begin{theorem}\label{Thm2}
For the ABS equations $\HD_\delta$, $\AT$ and $\QD_0$, the system (\ref{DaBT}) provides a non-natural auto-BT. In each case, the system (\ref{sbt}) provides a torqued auto-BT for $\HD_0$, and for $\AT$ and $\QD_0$ the system (\ref{abt}) provides an auto-BT for a quad-equation which depends on only lattice parameter. The non-natural auto-BTs and their torqued counterparts are provided in Table \ref{Tafel3}.
\begin{table}[H]
\begin{center}
\begin{tabular}{l|l|l|l}
\hline
Equation & Auto-BT & Symmetric t-auto-BT & Asymmetric t-auto-BT \\
$Q(p,q)=0$ & $Q(p,0)=Q(q,0)=0$ & $Q^t(p,0)=Q^t(q,0)=0$ & $Q^t(p,0)=Q(q,0)=0$ \\
\hline
$\HD_\delta$ & $\HD_\delta$ & $\HD_0$ & $-$\\
$\AT$ & $\AT$ \text{ or } $P:u\t u\h u \th u=1$ & $\HD_0$ & $\AT(0,q^{-1})$\\
$\QD_0$ & $\QD_0$ \text{ or } $K$ & $\HD_0$ & $\QD_0(0,q)$\\
\end{tabular}
\caption{\label{Tafel3} Degenerate auto-B\"acklund transformations.}
\end{center}
\end{table}
All the auto-BTs in Table \ref{Tafel3} admit (\ref{dspp}) as their superposition principle.
The symmetric torqued auto-BTs obtained from $\AT$ and $\QD_\delta$ furthermore admit the superposition principle
\begin{equation} \label{enilm}
(uw-1)(vz-1)=0.
\end{equation}
\end{theorem}

The quad equation $\HD_\delta(p,0)$ is of type H$^6$, a special case of D$_4$, with $\delta_1=\delta_2=0$, cf. \cite{Boll}. The equations
$\AT_\delta(p,0)$ and $\QD_0(p,0)$ are of type H$^4$, equivalent to a special case of $\HD$. By suitable M\"obius transformations, equation $D$ (\ref{eqD}) relates to the equation of type H$^6$ named D$_1$, and equations $K$ (\ref{eqK}) and $P$ relate to equation D$_4$, with $\delta_1=\delta_2=\delta_3=0$. The equations (\ref{dspp}) and (\ref{enilm}) are of type H$^4$, and relate to $\HE$, with $p=q=0$. The consistent cubes in Table \ref{Tafel3} which contain $P$ and $K$ do not have the tetrahedron property. All other consistent cubes in Table \ref{Tafel3} do have that property. However, we have not been able to identify any of them in the classification given in \cite{Boll,RB11}.

According to \cite[Definitions 5.1]{JH19} the non-natural auto-BTs for $\AT$ and $\QD_0$, as well as the symmetric t-auto-BTs for $\HE$ and $\HD_0$, are weak auto-BTs, as they give rise to multiple equations.

\section{A torqued $\HD^*_0$ equation}
A list of multi-quadratic quad equations: $\HT^*$, $\HD_\delta^*$, $\AE^*$, $\AT^*$, $\QE^*$, $\QT^*$, $\QD_\delta^*$, $\QV_c^*$, where $\delta\in\{0,1\}$ and $c\in\{0,\pm 1, \pm i\}$, was given in \cite[Section 4]{AtkN-IMRN-2014}. Due to the special factorisation of their discriminants these models can be reformulated as single-valued systems. The models can be consistently posed on the cube, however sign choices lead to four solutions. Here, we show that the torqued version of $\HD_0^*$ forms a symmetric auto-BT for a multi-quartic quad equation, and that it admits an asymmetric auto-BT itself. Whether other multi-quadratic quad equations admit torqued versions remains to be seen.

For the $\HD^*_\delta$ equation
\begin{equation}\label{H30*}
\begin{split}\HD^*_\delta([u,\h u],[\t u,\th u]; p, q)=&(p-q)\left(p(u\h u-\t u\th u)^2-q(u\t u-\h u\th u)^2\right)+pq(u-\th u)^2(\t u-\h u)^2\\&-4\delta^2(p-q)(u-\th u)(\t u-\h u)=0,
\end{split}
\end{equation}
we find that $\th v$ satisfies
\begin{equation}\label{H30*thb}
\begin{split}
&u\left(p(q-r)^2(\h u v-\th v\t u)+q(p-r)^2(\t u v-\th v\h u)+r(p-q)^2(\t u\h u-\th v v)\right)\\
&+2(\sigma_1\sigma_2-\delta^2)(r-p)(r-q)(\th v- v)
+2(\sigma_1\sigma_3-\delta^2)(q-p)(q-r)(\th v-\h u)\\
&+2(\sigma_2\sigma_3-\delta^2)(p-q)(p-r)(\th v-\t u)=0
\end{split}
\end{equation}
with
\begin{equation} \label{sigmas}
\sigma_1^2=pu\t u + \delta^2,~~\sigma_2^2=qu\h u+ \delta^2,~~\sigma_3^2=ru v+ \delta^2.
\end{equation}
Although equation \eqref{H30*thb} involves five points, when $\delta=0$ the variable $u$ can be eliminated (as $\sigma_i\sigma_j/u$ does not depend on $u$) and therefore the equation \eqref{H30*} with $\delta=0$ possesses the tetrahedron property.

The multiplicative torqued version of $\HD^{*}_0$ is defined as, see notation (\ref{mnot}),
\begin{equation}\label{tH30*}
\HD^{*m}_0([u,\h u],[\t u, \th u];\, p,q)=(p-1)(p(u\h u-\t u\th u)^2-(u\th u-\h u\t u)^2)+p(u-\t u)^2(\th u-\h u)^2=0
\end{equation}
with vanishing parameter $q$.
\begin{theorem}
The symmetric system
\begin{equation}\label{H30*BT}
\begin{split}
&\HD^{*m}_0([u,v],[\t u,\t v];\, p,r)=(p-1)(p(uv-\t u\t v)^2-(u\t v-v\t u)^2)+p(u-\t u)^2(v-\t v)^2=0,\\
&\HD^{*m}_0([u,v],[\h u,\h v];\, q,r)=(q-1)(q(uv-\h u\h v)^2-(u\h v-v\h u)^2)+q(u-\h u)^2(v-\h v)^2=0
\end{split}
\end{equation}
acts as an auto-BT for
\begin{equation}\label{H30*1}
\begin{split}
&~~~~~~~~~~~~~~~~~p(q-1)^2(\h uu-\th u\t u)+q(p-1)^2(\t uu-\th u\h u)+(p-q)^2(\t u\h u-\th uu) \\
&+2\sigma_2(p-q)(p-1)(\th u-\t u)+2\sigma_1(q-p)(q-1)(\th u-\h u)+2\sigma_1\sigma_2(1-p)(1-q)(\th u/u-1)=0,
\end{split}
\end{equation}
which is equivalent to a multi-quartic quad equation\footnote{This equation can be obtained by substituting (\ref{sigmas}) into
\[
Q(\sigma_1,\sigma_2)Q(\sigma_1,-\sigma_2)Q(-\sigma_1,\sigma_2)Q(-\sigma_1,-\sigma_2)=0,
\]
where $Q(\sigma_1,\sigma_2)$ is the left hand side of (\ref{H30*1}). It has 2293 terms, it is multi-quartic in $u,\t u,\h u, \th u$ (and homogeneous of total degree 8) and it is multi-octic in $p,q$ (with total degree 12).}, and the asymmetric system
\[
\HD^{*m}_0([u,v],[\t u,\t v];\, p,r)=\HD^{*}_0([u,v],[\h u,\h v];\, q,r)=0
\]
provides an auto-BT for $\HD^{*m}_0(p,q)[u]$.
\end{theorem}
\begin{proof}
If we take $\delta=0$ and switch $\t u \leftrightarrow \th u$, the equation (\ref{tH30*}) is obtained by replacing $p\rightarrow pq$ in (\ref{H30*}), while (\ref{H30*1}) is obtained by replacing $p\rightarrow pr, q\rightarrow qr, v\rightarrow u, \th v\rightarrow \th u $ in (\ref{H30*thb}). The symmetric claim follows from the tetrahedron property of $\HD^{*}_0$. Next, consider the cube system
\begin{equation}\label{tH30*cube}
\begin{split}
\HD^{*m}_0([u,\h u],[\t u,\th u];\, p,q)=0,~~~~\HD^{*m}_0([u,v],[\t u,\t v];\, p,r)=0,~~~~\HD^{*}_0([u,v],[\h u,\h v];\, q,r)=0,\\
\HD^{*m}_0([v,\h v],[\t v,\th v];\, p,q)=0,~~~~\HD^{*m}_0([\h u,\h v],[\th u,\th v];\, p,r)=0,~~~~\HD^{*}_0([\t u,\t v],[\th u,\th v];\, q,r)=0.
\end{split}
\end{equation}
In fact, the equation $\HD^{*m}_0([u,\h u],[\t u,\th u];\, p,q)=0$ can be written in the other patterns, i.e.
\begin{align}
&p(u\h u-\t u\th u)-(u\th u-\t u\h u)+2\sigma_1(\h u-\th u)=0, \label{th30*1}\\
   &p(u\h u-\t u\th u)+(u\th u-\t u\h u)+2\h{\sigma}_1(u-\t u)=0, \label{th30*2}
\end{align}
and for the case $\HD^{*m}_0([u,v],[\t u,\t v];\, p,r)=0$ we have
\begin{align}
&p(uv-\t u\t v)-(u\t v-\t uv)+2\sigma_1(v-\t v)=0, \label{th30uv*1}\\
   &p(uv-\t u\t v)+(u\t v-\t uv)+2\b{\sigma}_1(u-\t u)=0, \label{th30uv*2}
\end{align}
with $\b{\sigma}_1^2=pv\t v$.
Likewise, $\HD^{*}_0([u,v],[\h u,\h v];\, q,r)=0$ is equivalent to
\begin{equation}
   u(q(uv-\h u\h v)+r(u\h u-v\h v))-2(u-\h v)\sigma_2\sigma_3=0. \label{h30*1}
\end{equation}
From equation \eqref{th30*1}, \eqref{th30uv*1} and \eqref{h30*1} we derive
\begin{equation}\label{tH30*th}
\th u=\frac{pu+\t u+2\sigma_1}{p \t u+u+2\sigma_1}\,\h u,~~
\t v=\frac{pu+\t u+2\sigma_1}{p \t u+u+2\sigma_1}\,v,~~
\h v=\frac{quv+ru\h u-2\sigma_2\sigma_3}{qu\h u+ruv-2\sigma_2\sigma_3}\,u.
\end{equation}
Making use of equations \eqref{th30*2} and \eqref{th30uv*2}, we have
\begin{equation}\label{tH30*s1}
\h{\sigma}_1=-\frac{p(u\h u-\t u\th u)+u\th u-\t u\h u}{2(u-\t u)},~~
\b{\sigma}_1=-\frac{p(uv-\t u\t v)+u\t v-\t u v}{2(u-\t u)}.
\end{equation}
Substituting \eqref{tH30*th} and \eqref{tH30*s1} in
$\HD^{*m}_0([v,\h v],[\t v,\th v];\, p,q)=0,~\HD^{*m}_0([\h u,\h v],[\th u,\th v];\, p,r)=0$, i.e.
$$
p(v\h v-\t v\th v)-(v\th v-\t v\h v)+2\b\sigma_1(\h v-\th v)=0,~~p(\h u\h v-\th u\th v)-(\h u\th v-\th u\h v)+2\h\sigma_1(\h v-\th v)=0,
$$
we can determine
\begin{equation}\label{tH30*thb}
\th v=\frac{quv+ru\h u-2\sigma_2\sigma_3}{qu\h u+ruv-2\sigma_2\sigma_3}\,\t u.
\end{equation}
Based on \eqref{tH30*th} and the relation
$
p^2\sigma_2\sigma_3\t{\sigma_2\sigma_3}=qr\sigma_1^2\h{\sigma}_1\b{\sigma}_1
$,
solving $\HD^{*}_0([\t u,\t v],[\th u,\th v];\, q,r)=0$,
we derive the equivalent value
\[
\th v=\frac{2qru\h uv-\sigma_2\sigma_3(qv+r\h u)}{2qru\h uv-\sigma_2\sigma_3(q\h u+rv)}\,\t u.
\]
This shows that the cube system (\ref{tH30*cube}) is consistent.
\end{proof}

\section{Concluding remarks}\label{sec-5}
Given an auto-BT which is $u\leftrightarrow v$ symmetric, we have shown how to produce another auto-BT for a different equation. Thus, our method provides new relationships between equations and their auto-BTs. The crucial property which lies at the heart of the symmetric torqued auto-BTs is the tetrahedron property, see Figure \ref{tetrah}. The corresponding key property for the asymmetric torqued auto-BTs is the planar property, illustrated in Figure \ref{planar}. By torquing the natural ABS auto-BTs we obtained auto-BTs for an ABS equation of type $Q$, in the symmetric case, or a t-ABS equation, in the asymmetric case. Thus, each ABS equation has a torqued variant (additive or multiplicative):
\begin{align*}
\HE^a:\quad &(u-\t{u})(\h{\t{u}}-\h{u})-p=0,\\
\HT^a:\quad &(u-\t{u})(\h{\t{u}}-\h{u})-p(u+\t{u}+\h{u}+\h{\t{u}}+p+2q)=0,\\
\HD^m_\delta:\quad &p(u\h{\t{u}}+\h{u}\t{u})-(u\h{u}+\t{u}\h{\t{u}})+\delta(p^2-1)q=0,\\
\AE^a_\delta:\quad &(p+q)(u+\h{u})(\t{u}+\h{\t{u}})-q(u+\h{\t{u}})(\h{u}+\t{u})-\delta^2pq(p+q)=0,\\
\AT^m:\quad &p(1-q^2)(u\h{\t{u}}+\h{u}\t{u})-(1-p^2q^2)(u\h{u}+\t{u}\h{\t{u}})
-(p^2-1)q(1+u\t{u}\h{u}\h{\t{u}})=0,\\
\QE^a_\delta:\quad &(p+q)(u-\h{u})(\h{\t{u}}-\t{u})-q(u-\h{\t{u}})(\h{u}-\t{u})+\delta^2pq(p+q)=0,\\
\QT^a: \quad &(p+q)(u-\h{u})(\h{\t{u}}-\t{u})-q(u-\h{\t{u}})(\h{u}-\t{u})
+pq(p+q)(u+\t{u}+\h{u}+\h{\t{u}}-p^2+pq-q^2)=0,\\
\QD^m_\delta:\quad &p(1-q^2)(u\h{u}+\t{u}\h{\t{u}})-(1-p^2q^2)(u\h{\t{u}}+\h{u}\t{u})\\
& ~~~~ -(p^2-1)q\left(u\t{u}+\h{u}\h{\t{u}}+\frac{\delta^2(1-p^2q^2)(1-q^2)}{4pq^2}\right)=0,\\
\QV^a:\quad &\mathrm{sn}(p+q)(u\h{\t{u}}+\h{u}\t{u})-\mathrm{sn}(q)(u\h{u}+\t{u}\h{\t{u}})
+\mathrm{sn}(p)\Big(k\,\mathrm{sn}(p+q)\mathrm{sn}(q)(u\t{u}\h{u}\h{\t{u}}+1)
-\h{\t{u}}\h{u}-u\t{u} \Big)=0.
\end{align*}
Our results imply that each t-ABS equation is integrable in the sense that it is CAC, and hence it has an auto-BT, and admits a Lax-representation. The consistent cubes in which they appear are given in Tables \ref{Tafel1} and Table \ref{Tafel2}. For example, the second item in Table \ref{Tafel1} yields the following two consistent cubes, cf. Figure \ref{1a}:
\[
A = \HT^a(p,r) = 0, \qquad
B = \HT(q,r) = 0, \qquad
Q = \HT^a(p,q) = 0,
\]
where the auto-BT $A=B=0$ is asymmetric, and
\[
A = \HT^a(p,r) = 0, \qquad
B = \HT^a(q,r) = 0, \qquad
Q = \QE_1(p,q) = 0,
\]
where the auto-BT $A=B=0$ is symmetric. The torqued equations which appear in these cubes can be represented as in Figure \ref{ris}, and they can be extended to 4D cubes, cf. Figure \ref{4DC}, representing the superposition principles, which take the form $\HT([u,z],[v,w];\ r,s)=0$ for each of the above cases.

We have rediscovered all symmetric auto-BTs given in \cite[Table 2]{Atk08}. They are given by the second and third item in Table \ref{Tafel1} and the first two items in Table \ref{Tafel2} (right column). Thus we have shown how such auto-BTs arise from the natural auto-BTs which are implied by the CAC property of certain ABS equations. The torqued $Q$-equations were found before in \cite{ABS09}, where it was shown that their consistent cubes are equivalent to the natural one. The torqued H-equations are special cases of trapezoidal versions of H$^4$ equations classified in \cite{Boll,RB11}. In Table \ref{Tafel4} we relate our $H$-type consistent cubes to consistent cubes given in \cite{Boll}

\begin{table}[H]
\begin{center}
\begin{tabular}{l|l|l|l|l|l|l}
Table(row,column) & 1(1,1) & 1(1,2) $\QE_0$& 1(2,1) & 1(2,2) & 2(2,1) & 2(2,2) \\
\hline
Eq. in \cite{Boll} \phantom{${{1^1}^1}^1$} & (3.12) & (3.8) & (3.13) & (3.9) & (3.14) & (3.10) \\
\end{tabular}
\caption{\label{Tafel4} Torqued cube systems obtained from the H-equations identified as special cases of asymmetric systems derived  in \cite{Boll}. The element 1(1,1), in row 1, column 1 in Table 1, is $\HE^a$.}
\end{center}
\end{table}

In \cite{GSL,GY2} the trapezoidal H$^4$ and H$^6$ equations, which include the equations $D$, $K$, $P$, and $\HD_\delta(p,0)$, have been shown to be linearisable and Darboux integrable, cf. \cite{AS}. General solutions for those equations were obtained in \cite{GSY}.

The consistent cube in Table 1 with $D$, as well as the one in Table 2 with $K$, do not possess the tetrahedron property and hence fall outside the scope of \cite{Boll,RB11}. The non-natural auto-BTs for $\HD_\delta$, $\AT$ and $\QD_0$ were found in \cite{ZZ}. We have not been able to identify their 3D consistent cubes, nor those of their torqued variants, as presented in Table \ref{Tafel3}, in the classification \cite{Boll,RB11}. As these systems, with the exception of those carrying $P$ and $K$, do admit the tetrahedron property, this indicates that the classification \cite{Boll,RB11} is not complete. To the best of our knowledge, the torqued multi-quadratic equation (\ref{tH30*}) and the multi-quartic quad equation which can be derived from equation (\ref{H30*1}) are new.

By cyclic rotation or interchange of shifts and parameters, cf. \cite[Section 2.1]{JH19} or \cite[Lemma 2.1]{WVZ}, all equations which are part of an auto-BT admit an auto-BT themselves. This implies that each symmetric torqued auto-BT obtained in the current paper can be used to create an (asymmetric) auto-BT for the corresponding t-ABS equation which is different than the asymmetric torqued auto-BT provided here. In \cite{WVZ} such an auto-BT is explicitly given and subsequently used to find a seed and a 1-soliton solution for $\HT^a$. Similarly, auto-BTs for each equation in the auto-BTs in Table \ref{Tafel3} are obtained.

We would like to mention another mechanism which yields alternative auto-BTs, using parameter transformations. Let $A(p,q)=B(p,q)=0$ be an auto-BT for $Q(p,q)=0$. If a transformation $(p,q)\mapsto (p^\dagger,q^\dagger)$ leaves $Q$ invariant, then $A(p^\dagger,r)=B(q^\dagger,r)=0$ provides an alternative auto-BT for $Q(p,q)=0$. The ABS and t-ABS equations admit the parameter scaling symmetries presented in Table \ref{Tafel5}.

\begin{table}[H]
\begin{center}
\begin{tabular}{r|l}
$(p^\dagger,q^\dagger)$ & ABS and torqued ABS equations \\
\hline
$(p+c,q+c)$ & $\HE$ \phantom{${{1^1}^1}^1$}\\
\hline
$(cp,cq)$ & $\HD_0$,\quad $\QE_0$,\quad $\QE_0^a$ \phantom{${{1^1}^1}^1$}\\
\hline
\phantom{${{1^1}^1}^1$} $(-p,-q)$ & $\AE_\delta$,\quad $\AE_\delta^a$,\quad $\QE_\delta$,
\quad $\QE_\delta^a$,\quad $\QT$,\quad $\QT^a$,\quad $\QV$,\quad $\QV^a$ \\
\hline
$(p^{-1},q^{-1})$ & $\AT$,\quad $\AT^m$,\quad $\QD_\delta$,\quad $\QD_\delta^m$ \phantom{${{1^1}^1}^1$}
\end{tabular}
\caption{\label{Tafel5} Parameter scaling symmetries.}
\end{center}
\end{table}
We have considered auto-BTs which are $u\leftrightarrow v$ symmetric. It is also possible to torque (non-auto) BTs between two equations, however, in the resulting consistent cube system, equations on opposite faces may not be related by a shift on the lattice. For example, the BT between $\HE$ and $\HT$ \cite[Table 3]{Atk08} has the tetrahedron property. Performing a symmetric torque gives rise to
\begin{equation} \label{bta} \tag{62a}
\begin{split}
A=2u\t v - p - \t u - v,\quad &B=2u\h v - q - \h u - v \\
A^\ast=2\th u \h v - p - \h u - \th v,\quad &B^\ast=2\th u \t v - q - \t u - \th v
\end{split}
\end{equation}
which is not a BT between
\begin{equation} \label{btb} \tag{62b}
Q=(\h u - \t u)(u - \th u)+(p - q)(u + \th u)
\end{equation}
and
\begin{equation} \label{btc} \tag{62c}
Q^\ast=(\h v - \t v)(v - \th v)+(p - q)(\h v + \t v),
\end{equation}
as $A^\ast\neq \widehat{A}$ and $B^\ast\neq \widetilde{B}$. Instead, a BT is obtained on the black-white lattice, similar to \cite[Proposition 4.1]{XP}, cf. \cite[Remark 2.2]{ZVZ}. Other equations with BTs which have the tetrahedron or planarity property may exist and will be left for future investigation.

Finally, we note that multi-component t-ABS equations with corresponding auto-BTs exist, and can be constructed using the technique explained in \cite{ZVZ}.

\subsection*{Acknowledgements}
The authors are thankful for detailed comments and questions from the referees. Financial support was provided by a La Trobe University China studies seed-funding research grant,
by the department of Mathematics and Statistics of La Trobe University, by the NSF of China
[Grants 11631007, 11875040, 11801289], by the Science and technology innovation plan of Shanghai [No.20590742900], and by the K.C. Wong Magna Fund in Ningbo University.

\end{document}